\documentclass[12pt]{article}
\usepackage{eurosym}
\usepackage{amssymb}
\usepackage{amsfonts}
\usepackage{amsmath}
\usepackage{geometry}
\usepackage{setspace}
\usepackage{pdflscape}
\usepackage{comment}
\usepackage{graphicx}
\usepackage{caption}
\usepackage{subcaption}
\usepackage[colorlinks=true,allcolors=magenta]{hyperref}
\usepackage{natbib, bbm}

\newtheorem{theorem}{Theorem}

\newtheorem{corollary}{Corollary}

\newtheorem{definition}{Definition}

\newtheorem{lemma}{Lemma}

\newtheorem{proposition}{Proposition}

\newenvironment{proof}[1][Proof]{\noindent\textbf{#1.} }{\ \rule{0.5em}{0.5em}}
\input{tcilatex}
\begin{document}

\title{Learning, Diversity and Adaptation in Changing Environments: The Role
of Weak Links\thanks{%
We gratefully acknowledge financial support from the Hewlett Foundation, MIT
Quest, and the Mathworks Engineering Fellowship. We thank Ben Golub and Alex
Wolitzky for suggestions.}}
\author{ Daron Acemoglu \thanks{%
Department of Economics, Massachusetts Institute of Technology. Email: 
\texttt{daron@mit.edu}.} \and Asuman Ozdaglar\thanks{%
Department of Electrical Engineering and Computer Science, Massachusetts
Institute of Technology. Email: \texttt{asuman@mit.edu}.} \and Sarath
Pattathil\thanks{%
Department of Electrical Engineering and Computer Science, Massachusetts
Institute of Technology. Email: \texttt{sarathp@mit.edu}.} }

\date{}
\maketitle

\begin{abstract}
Adaptation to dynamic conditions requires a certain degree of diversity. If
all agents take the best current action, learning that the underlying state
has changed and behavior should adapt will be slower. Diversity is harder to
maintain when there is fast communication between agents, because they tend
to find out and pursue the best action rapidly. We explore these issues
using a model of (Bayesian) learning over a social network. Agents learn
rapidly from and may also have incentives to coordinate with others to whom
they are connected via strong links. We show, however, that when the
underlying environment changes sufficiently rapidly, any network consisting
of just strong links will do only a little better than random choice in the
long run. In contrast, networks combining strong and weak links, whereby the
latter type of links transmit information only slowly, can achieve much
higher long-run average payoffs. The best social networks are those that
combine a large fraction of agents into a strongly-connected component,
while still maintaining a sufficient number of smaller communities that make
diverse choices and communicate with this component via weak links.


\end{abstract}

\bigskip \thispagestyle{empty}

\bigskip

\bigskip

\bigskip

\bigskip

\bigskip

\bigskip

\bigskip

\thispagestyle{empty}\bigskip \restoregeometry

\newpage

\setcounter{page}{1}

\section{Introduction\label{sec:intro}}

Rising threats from economic disruptions, climate change, new pandemics and
resurgent nationalism and other extremist ideologies have rekindled interest
in understanding what makes societies resilient against challenges (\cite%
{Keck, brunnermeier2022resilient}). A large body of literature in ecology
and biology, starting with the influential work of \cite{fisher1958genetical}%
, suggests that diversity is critical for resilience in the face of changing
circumstances. Species that lack diversity may be well-suited for a given
environment, but then have a hard time adapting to sizable changes. \ In
this paper, we explore the relationship between diversity and adaptation in
a social context. Our focus is on one facet of this problem: learning about
and adapting to a changing environment.

\subsection{Main Argument}

We envisage a set of agents that interact with each other and choose between
two actions, one of which has higher payoffs. Local interactions and
knowledge flows create a force towards choosing the same action as one's
neighbors. However, when all or most agents choose the same action, even if
this has currently the higher payoff, learning about dynamically-improving
alternatives becomes more difficult and may reduce long-run payoffs.

Formally, we study the Bayesian-Nash equilibrium of a simple game in which
each agent has a utility consisting of a \emph{material payoff}, which
depends on whether her action matches the underlying state, and a \emph{%
network payoff}, which depends on how well her action matches the actions of
her closely-associated neighbors. Crucially, the underlying state changes
over time according to a Markov chain, necessitating adaptation to this
evolving environment. We assume that agents use Bayesian updating to form
their beliefs about the underlying state, but in our baseline model simply
maximize their current utility, and hence have no reason to experiment.
(These results are then extended to the case in which agents maximize their
discounted utility).

The need for adaptation in our model creates a network version of the
classic exploitation-exploration trade-off: how much should some agents
deviate from what is best and experiment to see whether the underlying
environment has changed? Differently from standard experimentation problems,
however, here the network structure is critical. If agents are closely
linked together, they tend to play the same action, both because of local
information flows and because of local interactions. But in the case where
all agents play the best action today, there is an adaptation problem: when
the environment changes and there is a need for adaptation---switching to
the now-higher payoff action---such a change does not take place or does so
very very slowly. If, on the other hand, the social network has several
disconnected components, some of which are playing diverse actions, the
society as a whole will discover a change in environment rapidly, but this
information will not be transmitted from a local community to the rest of
society, because of the disconnected nature of the social graph.

Our main argument in this paper is that Granovetter's idea of
\textquotedblleft weak links\textquotedblright\ (\cite%
{granovetter1977strength}), which do not have the same frequency of close
interaction but can act as occasional conduits of information, provides a
powerful solution to this problem. Building on this idea, we model weak
links as intermittently transmitting information about behavior and payoffs,
but without inducing locally-uniform actions. We prove that a society
consisting of several clusters that are strongly disconnected but weakly
linked can achieve fast adaptation to changing circumstances, while ensuring
that most agents take the high-payoff action most of the time.

We also characterize the best social network from the viewpoint of
maximizing average long-run payoff. A version of the star network turns out
to be the one that achieves the highest average long-run payoff. This star
network involves a large number of agents strongly-linked clustered in a
star-like node in the middle, with a sufficient number of small communities
that are the weakly-connected leaves of this star node. The leaves do the
experimentation and ensure that society as a whole quickly learns when the
underlying environment changes. The star-like community in the middle
exploits both the gains from local interactions and the information
benefits, which it quickly acquires from the leaves.

We confirm that weak links are critical for this result by showing that
without weak links average long-run payoffs are only a little better than
random play, because most agents stay stuck with actions that were once good
but\ have since ceased to be so. In such networks, adaptation to changing
environment comes only from slow mutations/mistakes, and under our
assumptions, the rate of such switches is much slower than the rate at which
the environment changes. This result clarifies that it is the presence of
weak links, with an appropriate topology of strong connections, that ensures
that society achieves approximately the highest possible payoff.

Our basic analysis is for the case in which agents choose the action that
maximizes their current payoff. We additionally show that our results extend
to the case in which agents maximize discounted payoffs, provided that their
discount rate is not too high. Specifically, we derive a bound on this
discount rate such that below this bound, all of our main results continue
to hold, and in particular, without weak links, average long-run payoffs are
approximately as good as random play, while star-like networks can achieve
much higher long-run average payoffs. This bound depends on the strength of
local interactions and the maximum degree of the network.

\subsection{Broader Context}

We view our results as relating not just to the game theory and economics
literatures, but also to the broader literature on diversity and adaptation.
The theme of diversity is central in biology, but without the key issue that
arises in social systems: incentivizing agents to take actions that will
preserve diversity.

The adaptation benefits of diversity receive support from studies of several
different species. For example, \cite{agha2018adaptation} demonstrate
experimentally that cyanobacteria are much more vulnerable to a fungal
parasite when they are homogeneous. In fact, in host populations that are
kept homogeneous, parasites can spread very rapidly, whereas genetically
diverse host populations can resist the parasite much more successfully,
because they contain genes that are less vulnerable to the specific parasite
and these genes multiply faster in response to invasion. Similar benefits of
diversity are observed among bees in response to fluctuations in
temperature, as shown in \cite{fischer2004bee}. Each individual bee's
temperature thresholds for huddling and fanning are tied to a genetically
linked trait. Hives that lack genetic diversity in this trait experience
unusually large fluctuations in internal temperatures whereas hives with
genetic diversity produce much more stable internal temperatures. Thus, the
genetic diversity of the bees leads to relatively stable temperatures that
ultimately improve the health of the hive.

Even in biological systems, maintaining diversity is a major challenge. One
of the most widely-held theories of the benefits of sexual reproduction is
precisely that it ensures sufficient diversity within both organisms and
populations by mixing alleles from the two parents (see for example \cite%
{Weismann, barton1998sex, burt2000perspective}). As a result, sexual
reproduction enables greater fitness via adaptation to changing environments
relative to asexual reproduction. Experiments on yeast provide evidence for
this hypothesis. In particular, \cite{Goddard} genetically modified a strain
to create two strands of yeast that are identical, except for the way they
reproduce, and confirmed that the sexually-reproduced strand was much more
adaptable to harsher environments than the one that reproduces asexually.

Similar adaptation benefits of diversity have been hypothesized in social
settings and sometimes documented. \cite{granovetter1977strength,
Grano_network, granovetter2017society} have argued that new superior
technologies spread rapidly in tech clusters, such as Silicon Valley, via
weak links, that were created either by communication between employees or
managers of different companies or directly by workers moving between
companies (see \cite{saxenian1996regional}, on this pattern in Silicon
Valley, and \cite{jacobs2016economy}, for a more general emphasis on this
aspect of communication in urban environments).\footnote{\cite%
{rajkumar2022causal} confirm using Linkdin data that weak links are still
central for job finding. They emphasize that there is an \textquotedblleft
inverted U-shaped relationship between the weak tie strength and job
transmission such that weaker ties increased job transmission but only to a
point, after which there were diminishing marginal returns to tie
weakness.\textquotedblright\ This is in line with the results in our
Proposition \ref{prop:island_weak}.} Other studies emphasize the importance
of agents that bridge \textquotedblleft structural holes\textquotedblright\
between different parts of a community (\cite{burt1992structural}). This
perspective also provides a reinterpretation of the concerns articulated by
Robert Putnam (\cite{putnam2000bowling}) due to the declining importance of
diverse organizations, such as bowling alleys, sports clubs and local
religious organizations, which can provide the type of weak link that bridge
structural holes and communicate information between distinct social groups
that otherwise seldom interact. Our context also emphasizes that it is
particularly important that this takes place without creating the powerful
tendency towards homogeneity that strong links tend to induce. \ 

\subsection{Economics and Game Theory Literatures}

Within the economics and game theory literatures, our paper is related to a
number of distinct literatures. The first is a small literature on
adaptation and diversity. \cite{gross1996alternative} studies the reasons
why there is large nontypical variation within species and links this to
adaptation. More closely related is \cite{santos2008social} who analyze the
role of diversity in public good games and argue that diversity promotes
cooperation. The general presumption in much of economics is that diversity
in modern societies is conducive to conflict (e.g., see the survey in \cite%
{la2006racial}), though a few papers, such as \cite{montalvo2021ethnic},
document various economic benefits from diversity as well.

Several papers in economics study learning dynamics over social networks.
Our work is most directly related to the branch that focuses on Bayesian
models, such as \cite{gale2003bayesian}, \cite{banerjee2004word}, \cite%
{smith2008rational}, \cite{callander2009wisdom}, and \cite%
{acemoglu2011bayesian}. In addition, several papers, most notably \cite%
{bala1998learning, bala2001conformism}, \cite{demarzo2003persuasion} and 
\cite{golub2010naive} discuss non-Bayesian learning over social networks.
None of these papers consider the problem of adaptation to changing
environments, though the issue of balancing conformity from strong linkages
vs. sufficient incentives for agents to take different actions comes up in 
\cite{smith2008rational} and \cite{acemoglu2011bayesian}. More closely
related are a few recent papers that consider the speed of learning in
related problems. For example, \cite{acemoglu2022learning} characterize the
speed of learning with Bayesian agents observing different samples of past
online reviews, and we refer the reader to their paper for a discussion of
speed of learning results in the literature.

Even more closely related to our work are a few papers studying learning
when the underlying state is changing. \cite{moscarini1998social} observe
that, unless the underlying state is `sufficiently persistent', there cannot
be (Bayesian) cascades on a single action. \cite{frongillo2011social}
consider various non-Bayesian learning rules and show that they converge to
a steady-state distribution on complete graphs, despite the changing
environment. \cite{dasaratha2018learning} study a learning model where
individuals learn from others and their own private signals, and show that
learning is improved when private signals are diverse, which has a related
logic to our main results. Finally, \cite{levy2022stationary} is also
closely related, as they note that, with symmetric agents, all players
rapidly converge to the same (consensus) action, even after the underlying
state changes. None of these papers, nor any others that we are aware of,
study Bayesian learning under a general network and a changing state;
characterize which types of networks lead to better learning performance; or
model and observe the importance of weak links.

The structure of our model is also connected to the literature on
evolutionary or learning dynamics and equilibrium selection. Within this
literature, the pioneering work by \cite{kandori1993learning} consider an
evolutionary model with a finite number of agents randomly matching and
playing a two-player coordination game, subject to noise or mutations. They
show that the presence of noise reduces the range of long-run
\textquotedblleft equilibria\textquotedblright\ (stable configurations), and
in particular, in a $2\times 2$ game, evolutionary dynamics lead towards the
Pareto dominant Nash equilibrium. In related work, \cite{young1993evolution}
characterizes the stochastically stable equilibria in a large finite
population game subject to random matching and noise. As in \cite%
{kandori1993learning}, noise acts as an equilibrium selection device. \cite%
{ellison1993learning} points out that equilibrium selection in\ \cite%
{kandori1993learning} and \cite{young1993evolution} is very slow and
suggests that local matching---rather than random matching---leads to
significantly faster convergence. There are several important differences
between our work and this literature. First, to the best of our knowledge,
issues of adaptation to a changing environment or the role of diversity are
not studied in this literature. Instead, this literature's focus has been on
equilibrium selection in games with multiple equilibria. Second, rather than
evolutionary rules or rule-of-thumb behaviors, we focus on Bayesian-Nash
equilibria of a game with a changing underlying state.

Finally, some of the mathematical methods we use are common with the
literature on general belief dynamics. \cite{holley1975ergodic}, for
example, study the so-called \textquotedblleft voter model\textquotedblright
, which is similar to the evolutionary dynamics in \cite{kandori1993learning}
and \cite{young1993evolution} based on random matching (whereby influence
flows within the randomly-matched pair). In contrast, the stochastic
dynamics that emerge from our model is more similar to the \textquotedblleft
majority dynamics\textquotedblright\ studied in \cite{kanoria2011majority}
and \cite{yildiz2010voting}.

\subsection{Rest of the Paper}

The rest of the paper is organized as follows. In Section \ref{sec:model} we
introduce the model and define Bayesian-Nash equilibria. In Section \ref%
{sec:equi_char} we characterize the equilibria and provide a method to
analyze it for general networks. In Section \ref{sec:under_res}, we compare
networks with strong and weak links, and analyze networks which provide the
highest welfare. In Section \ref{sec:forward_looking}, we extend these
results to forward looking agents, and finally we provide a discussion of
our results in Section \ref{sec:discussion}.

\section{Model}

\label{sec:model}

In this section we introduce the basic environment, describe the network
formed by strong and weak links, payoffs, average welfare, and define
Bayesian-Nash equilibria.

\subsection{Network}

We consider a set of agents $V=\{1,2,\cdots ,n\}$ represented by nodes in an
undirected graph $G$. There are two kinds of links, \emph{strong }and \emph{%
weak}. We represent strong links with the symmetric matrix $S^{G}\in
\{0,1\}^{n\times n}$, with the convention that 
\begin{equation*}
S_{ij}^{G}=%
\begin{cases}
1 & \text{if agents $i$ and $j$ have a strong link between them} \\ 
0 & \text{otherwise}%
\end{cases}%
\end{equation*}%
The neighborhood of an agent $i$ is defined with respect to strong links, as 
${N}^{G}(i)=\{j\in V:S_{ij}=1\}$. The maximum degree of the network is
denoted by $d_{\max}^{G}=\max_{i\in V}|{N}^{G}(i)|$.

Weak links, on the other hand, are described by the symmetric matrix $%
W^{G}\in \{0,1\}^{n\times n}$, where similarly%
\begin{equation*}
W_{ij}^{G}=%
\begin{cases}
1 & \text{if agents $i$ and $j$ have a weak link between them} \\ 
0 & \text{otherwise}%
\end{cases}%
\end{equation*}%
We also let $\mathcal{E}^{G}_{s}$ and $\mathcal{E}^{G}_{w}$ denote the set
of strong and weak links respectively. Whenever this will cause no
confusion, we drop the superscript $G$.

\subsection{Actions and Rewards}

Time is continuous and runs to infinity. At each time $t\in \lbrack 0,\infty
)$, agent $i\in V$ chooses an action $a_{i}(t)\in $ $\mathcal{A}=\{0,1\}$.
The agent's resulting payoff is the sum of two components:

\begin{enumerate}
\item a \emph{material payoff} $R_{a}(t)$, which only depends on the action $%
a$ taken by the agent and the underlying state of nature (and is thus
stochastic);

\item a \emph{network payoff}, which depends on actions in the agent's
neighborhood as we describe in Section \ref{subsec:agent_dyna}.
\end{enumerate}

The need for adaptation arises because the underlying state and thus the
material payoffs from the two actions, $R_{0}(t)$ and $R_{1}(t)$, change
over time. We assume that these changes arrive according to a Poisson clock
of rate $\lambda $, and denote the (random) instances at which such changes
take place by $\{T_{k}\}_{k=0}^{\infty }$ and and refer to them as \textit{%
payoff shocks} (and we set $T_{0}=0$). Without loss of any generality, we
assume that following the realization of the Poisson clock at time $T_{k}$,
rewards change at $T_{k}^{+}$, that is, right after $T_{k}$. This implies
that the rewards from an action are constant over $(T_{k},T_{k+1}]$ for all $%
k$. We also simplify our analysis by assuming that the gap between the two
actions, $R_{0}(t)$ and $R_{1}(t)$, is constant and normalize it to 1,
though, crucially, which action has higher payoff naturally changes with the
realizations of the Poisson clock. We additionally define $A(t)$ as the
action with the higher reward at time $t$, and denote by $%
\{A_{k}\}_{k=1}^{\infty }$ the action with the higher reward in the time
interval $(T_{k-1},T_{k}]$.

Summarizing this reward structure, we can write that for all (random) time
instances $\{T_{k}\}_{k=0}^{\infty }$, we have%
\begin{equation*}
R_{0}(T_{k}^{+})-R_{1}(T_{k}^{+})=%
\begin{cases}
+1 & \text{w.p. }1/2 \\ 
-1 & \text{w.p. }1/2%
\end{cases}%
\end{equation*}%
with $R_{0}(t)-R_{1}(t)=R_{0}(T_{k}^{+})-R_{1}(T_{k}^{+})\ $for all$\ t\in
(T_{k},T_{k+1}]$.

Note also that the case where $\lambda =0$ yields the special case where
material payoffs are constant and known. We assume that all agents are
initialized (at time $t = 0$) to play Action 0.

\subsection{Information Structure\label{subsec:strong_weak_links}}

We next describe the information structure, which depends on the nature of
strong and weak links.

\textbf{Strong Links: }At all times $t$, each agent $i$ will have complete
information about the action history and associated payoffs from its
strongly-linked neighbors in the set $\mathcal{N}(i)$.

\textbf{Weak Links:} In contrast to strong links, weak links transmit
information slowly. We model this by assuming that weak links start as
\textquotedblleft dormant \textquotedblright\ and are activated
stochastically. Specifically, there is a Poisson clock of rate $\gamma $,
and each time the clock ticks, one dormant weak link is activated.
Furthermore, once a weak link is activated, it transmits information, and
then goes to an \textquotedblleft inactive\textquotedblright\ state until
another independent Poisson clock, this time of rate $\phi $, turns it back
to \textquotedblleft dormant \textquotedblright . We explain below the
reasoning for this two-stage activation. We first explain how the activated
weak link is chosen from the set of all weak links.

Let $\mathcal{W}(t)=\left\{ (i,j)\in \mathcal{E}_{w}^{G}:a_{i}(t)\neq
a_{j}(t),(i,j)\text{ is dormant}\right\} $. In other words, this is the set
of weak links that are dormant and also involve two linked agents playing
different actions at time $t$. This is the set of weak links that are
relevant for information transmission---since there is no relevant
information to be transmitted between agents that are playing the same
action. We\ assume that, once the relevant Poisson clock clicks, a link is
chosen uniformly at random from $\mathcal{W}(t)$. Once this happens, the
link becomes active, and information transmission happens through this link,
i.e., if the link that is fully activated is $(i,j)$, then the current
action and payoff of individual $i$ is transmitted to $j$, and symmetrically
information from $j$ is observed by $i$. Once this information has been
transmitted, the link enters an inactive state, in which it stays till the
Poisson clock of rate $\phi$ clicks, after which it becomes dormant again.

A couple of comments are useful at this point. First, information
transmission on weak links is slow, in contrast to the very fast
transmission on strong links. While strong links capture frequent
interactions, such as between family members, coworkers or closely-connected
agents, weak links transmit information occasionally via gossip or random
observation. In terms of our mathematical formulation, a weak link transmits
information only after moving from inactive to the dormant state, and then
waiting to become active. This slow transmission plays a key role in our
results, as we will see. Second, the fact that activated weak links are
among those connecting agents playing different actions is consistent with
the idea that weak links become active for gossip or information exchange.
The main reason this assumption is imposed in our setting is for simplicity:
without this assumption, some of the weak links that are activated would not
transmit relevant information, and although this does not affect our general
results, having activated links that do not transmit useful information
makes the coupling arguments we use for the proofs more difficult. Third,
the two-stage activation is important to ensure sufficient slowness in
information transmission. In particular, if there was no inactive state, it
might be the case than the same weak link could be chosen multiple times
(since weak links are selected from those playing different actions) while
other weak links are never activated. With our two-stage activation, we
ensure that once a weak link transmits relevant information, it moves to an
inactive state, where it is unable to transmit any information for a certain
\textquotedblleft backoff" period, dictated by the Poisson clock of rate $%
\phi $, after which it becomes dormant, where it is a contender to become a
conduit of information. In this formulation, $\phi \rightarrow \infty $
corresponds to the case where weak links are never in the inactive state,
whereas $\phi \rightarrow 0$ corresponds to the case where after a weak link
is activated to transmit information, it enters an inactive state forever
and will never again transmit information.

\subsection{Overall Payoffs and Beliefs}

\label{subsec:agent_dyna}

Agents maximize their static, current payoffs (until Section \ref%
{sec:forward_looking}, where we introduce forward-looking behavior). As
noted above, the overall per-period utility of an agent $i$ taking action $%
a_{i}$ at time $t$ is given by 
\begin{equation*}
\mathcal{U}_{i}^{a_{i}}(t)=R_{a_{i}}(t)+\tau f_{i}(a,t),
\end{equation*}%
where $R_{a}(t)$ is this agent's \emph{material payoff}, as specified above,
while\emph{\ }$\tau f_{i}(a,t)$ is her \emph{network payoff}, with $%
a=[a_{1},a_{2},\cdots $ $\cdots ,a_{N}]$ denoting the entire action profile
of this population (though what matters will be the actions of agent $i$'s
neighbors). Specifically, we equate this network payoff with the the number
of agent $i$'s neighbors playing action $a_{i}$ at time $t$. That is,%
\begin{equation*}
f_{i}(a,t)=\sum_{j\in {N}^{G}(i)}\mathbb{I}_{a_{j}(t)=a_{i}},
\end{equation*}%
where $\mathbb{I}_{a_{j}(t)=a_{i}}$ is the indicator function for neighbor $%
j $ of agent $i$ taking the same action $a_{i}$ is this agent at time $t$.
Intuitively, this term captures the payoff benefits from coordinating with
closely connected agents. The parameter $\tau \geq 0$ designates the
importance of this local network payoff.\footnote{%
All of our results in this section remain valid when $\tau =0$, so that
there is no network payoff, but such local payoff interactions become
important in the forward-looking case, analyzed in Section \ref%
{sec:forward_looking}.}

While the network payoff is deterministic (given an action profile of other
agents), the material payoff is stochastic and depends on the underlying
state, as specified above. Hence, agent best responses will depend on their
beliefs, which we next describe.

Let $\mu _{i}(t)$ denotes the belief of agent $i\in V$ that Action $1$ has
higher reward at time $t$, i.e., $R_{1}(t)-R_{0}(t)=+1$. More formally, 
\begin{equation*}
\mu _{i}(t)=\mathbb{E}_{i,t}[\mathbb{I}_{R_{1}(t)>R_{0}(t)}],
\end{equation*}%
where $\mathbb{E}_{i,t}$ denotes expectations according to the {information
set} of agent $i$ at time $t$, and $\mathbb{I}_{R_{1}(t)>R_{0}(t)}$ is the
indicator function for $R_{1}(t)>R_{0}(t)$. We assume that for all agents $i
\in V$, we have $\mu_i(0) = 1/2$, i.e., the agents have no information at
time $t =0$ about which action has the higher material payoff.

The assumption that agents maximize their current payoffs implies that%
\begin{equation}
a_{i}(t)=\underset{a\in \{0,1\}}{\func{argmax}}\ \mathbb{E}_{i,t}[\ \mathcal{%
U}_{i}^{a}(t)\ ].  \label{eq:action_choice}
\end{equation}

Finally, as in \cite{kandori1993learning} and \cite{young1993evolution}, we
introduce individual trembles. We assume that another Poisson clock of rate $%
\epsilon >0$ induces change in behavior. In particular, each time this clock
ticks one agent is picked uniformly at random and she ends up taking the
opposite action to the one she intended. \ We refer to this phenomenon as an 
\textit{$\epsilon $-tremble}. Throughout, we will take $\epsilon $ to be
small, and in fact much smaller than the rate at which the underlying state
changes ($\lambda $) and weak links transmit information ($\gamma $).

\subsection{Bayesian-Nash Equilibrium}

We focus on the Bayesian-Nash equilibria of this game. A Bayesian-Nash
equilibrium (BNE) is defined in a standard fashion.

\begin{definition}[Bayesian-Nash Equilibrium]
\label{def:BNE} An action profile $a=[a_{1},a_{2},\cdots $ $\cdots ,a_{N}]$
where $a_{i}\in \{0,1\}$ is a pure strategy BNE if for each $i$, $a_{i}$
maximizes the expected payoff in equation \eqref{eq:action_choice}, given
the action profile of other agents $a_{-i}$, with the expectation in %
\eqref{eq:action_choice},$\ \mathbb{E}_{i,t}$, taken according to Bayes rule.%
%
%
%
%
%
%
%
%
\end{definition}


\subsection{Average Welfare}

We evaluate the adaptation success of different social networks by looking
at their long-run average payoff (in BNE). This measure is attractive
because only societies that rapidly respond to a changing environment can
achieve high long-run average payoffs.\footnote{%
If instead we focused on discounted average payoffs, this would down-weight
future failures to adapt to changes.} Formally, average payoffs in a society
comprised of $n$ agents at time instance $T_{k}$ is%
\begin{equation*}
\mathcal{S}_{k}^{{G}}=\frac{1}{n}\sum_{i=1}^{n}\mathbb{I}%
_{a_{i}(T_{k})=A_{k}},
\end{equation*}%
where $\{T_{k}\}_{k=0}^{\infty }$ are instances of payoff shock and the
indicator function $\mathbb{I}_{a_{i}(T_{k})=A_{k}}$ takes the value $1$
when $a_{i}(T_{k})=A_{k}$ and $0$ otherwise. We condition on the social
network designated by graph $G$. Long-run average welfare is then defined as:%
\begin{equation}
\mathcal{S}^{{G}}=\lim_{k\rightarrow \infty }\mathcal{S}_{k}^{{G}}.
\label{def:average_welfare}
\end{equation}%
%
%
%
%
%
%
%
%
A couple of points are worth noting. First, we focus only on instances of
payoff shock, since in between payoff shocks nodes are (potentially) in a
transient state, trying to learn through strong and weak links about which
action has the higher material payoff. Second, we could have equivalently
defined long-run average payoffs as the average across all time periods.
This alternative definition depends on initial actions, though the weight of
these initial actions goes to zero as the limit is taken. The current
definition simplifies the exposition without loss of generality.


\section{Equilibrium Characterization}

\label{sec:equi_char}

In this section, we characterize the BNE and then provide an expression for
average welfare in any BNE. Our characterization proceeds as follows. First,
we prove a monotonicity property of Bayesian beliefs, establishing that
belief dynamics before the next time of information arrival never reverse
direction and they jump to the correct probabilities at times of information
arrival. Using this characterization, we prove that an agent will only
change her action during times of information arrival. Combining this result
with the structure of strong links, we show that, except at times of
information arrival, strongly-linked components will always play the same
action in any BNE. In the last subsection of this section, we provide a
characterization of average welfare under this equilibrium structure, using
a suitably designed embedded Markov chain, defined over the action profiles
of agents in the social network.

\subsection{Belief Dynamics}

The next definition introduces the (set of) times of information arrival.
Intuitively, these are time instances for an agent $i$ during which the
agent receives \textquotedblleft new information\textquotedblright . This
can happen because a weak link adjacent to this agent is activated, or a
strongly-linked neighbor changes her behavior, or the agent herself has an $%
\epsilon $-tremble. Formally:

\begin{definition}[Last time of new information]
\label{def:Last_New_Info_Time} Instance $t$ is a time of information arrival
for agent $i$ if one of the following take place at time $t$:\ 

\begin{itemize}
\item A weak link adjacent to agent $i$ is activated.

\item For some $j\in N(i)$, we have $a_{i}(t)\neq a_{j}(t)$.

\item Agent $i$ has an $\epsilon $-tremble.
\end{itemize}

Times of information arrival for agent $i$ are then defined as%
\begin{equation*}
\mathcal{T}_{i}=\{t:\text{ $t$ is a time of information arrival for agent $i$
}\},
\end{equation*}%
and the last instance of information arrival before $t$ is%
\begin{equation*}
T_{i}(t)=\sup \{t_{i}:t_{i}\in \mathcal{T}_{i},\text{ and }t_{i}\leq t\}.
\end{equation*}
\end{definition}

We remind the reader that, given the structure of information specified so
far, all instances of information arrival are fully-revealing about which
action has the higher (material) reward. Hence, \ agent $i$'s belief that
action $1$ is the better action at a time of information arrival is either 0
or 1.

We also note that agent $i$'s information set at time $t$, denoted by $%
\mathcal{I}_{i}(t)$, is fully summarized by the last instance of information
arrival before time $t$, $T_{i}(t)$, and the action profile observed by the
agent at this point. \ Recall that $\mu _{i}(t)$ is agent $i$'s belief that
action $1$ has greater material payoff at time $t$ than action $0$, and thus 
$\mu _{i}(t)=\mathbb{P}(A(t)=1\ |\ \mathcal{I}_{i}(t))$ (where also recall
that $A(t)$ denotes the action that has higher material payoff at time $t$,
which is common across all agents).

Using this notation, we can now establish a critical property of Bayesian
updates, which will enable us to characterize BNE.

\begin{lemma}
\label{lemma:bayes_mono} Bayesian beliefs at time $t$, $\mu _{i}(t)$,
satisfy the following monotonicity property:%
\begin{align}
\mu _{i}(T_{i}(t))=1& \implies \mu _{i}(t)>\frac{1}{2}  \notag \\
\mu _{i}(T_{i}(t))=0& \implies \mu _{i}(t)<\frac{1}{2}.  \notag
\end{align}
\end{lemma}

Lemma \ref{lemma:bayes_mono} states that once an agent becomes aware of the
action with the higher material payoff (which takes place following a time
of information arrival), her beliefs remain that this action is more likely
to be the higher-reward action until the next instance of information
arrival. Consequently, once an agent believes that, say, action $1$, is
better at time $T_{i}(t)$, then she will continue to believe that action $1$
is better than action $0$ ($\mu _{i}>1/2$) until she receives new
information.

While Lemma \ref{lemma:bayes_mono} establishes monotonicity of Bayesian
beliefs, it does not provide a full characterization of belief dynamics.
Such a characterization is difficult in general, though it can be obtained
in some special cases, as the next example shows. This example is included
purely for illustrative purposes, and in the rest of the paper we only use
the monotonicity result in Lemma \ref{lemma:bayes_mono}.

\begin{description}
\item \label{ex:belief_exact}

\item[Example 1] We now provide a special case of our model in which there
is enough symmetry in the network that Bayesian updates can be explicitly
characterized and, of course, verifies Lemma \ref{fig:Island_network}. In
this example, we first impose some restrictions on the graph structure.
Specifically, we assume that the graph $G$ satisfies 
\begin{equation*}
S_{ij}^{G}=1,S_{ik}^{G}=1\ \implies S_{jk}^{G}=1\qquad \forall \ i,j,k\in V
\end{equation*}%
and 
\begin{equation*}
W_{ij}^{G}=1\qquad \forall \ i,j\in V\ \text{such that }S_{ij}^{G}=0.
\end{equation*}%
We have just defined a class of networks where there is a weak link between
any two agents not connected by a strong link and agents form islands of
strongly linked cliques, i.e., each agent is part of a clique of strongly
connected agents. There are several such cliques, forming \textquotedblleft
islands\textquotedblright . Figure \ref{fig:Island_network} provides an
illustration of such a network. 
\begin{figure}[tbp]
\centering
\includegraphics[width=0.45\textwidth]{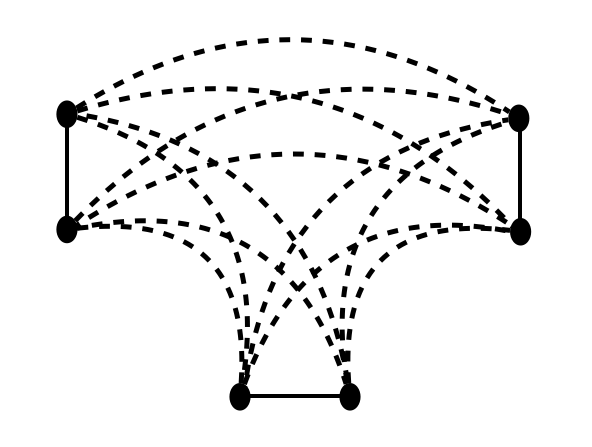}
\caption{A network which satisfies the assumptions in Example 1}
\label{fig:Island_network}
\end{figure}
Next, we also restrict the weak link structure further, by assuming that
each time the Poisson clock (of rate $\gamma $) ticks all weak links in the
network are activated. Under these assumptions, and as we take the limit $%
\epsilon \rightarrow 0$, we can characterize the exact belief updates as
follows 
\begin{equation*}
{\mu }_{i}(t)=(1-e^{-\lambda (t-T_{i}(t))})\frac{1}{2}+e^{-\lambda t}{\mu }%
_{i}(T_{i}(t)).
\end{equation*}%
It is also straightforward to see that the beliefs in this equation satisfy
the monotonicity property in Lemma \ref{lemma:bayes_mono}. \hfill
\end{description}

%

Belief monotonicity in Lemma \ref{lemma:bayes_mono} immediately yields our
next result, which shows that agents only change their action during times
of information arrival.

\begin{lemma}
\label{prop:agent_switch}For any agent $i\in V$ and all $t$, we have:%
\begin{equation*}
a_{i}(t^{-})\neq a_{i}(t)\implies t\in \mathcal{T}_{i}.
\end{equation*}
\end{lemma}

Hence, if any agent changes her action at time $t$, then it must be the case
that $t$ is a time of information arrival. A direct but important
consequence of Lemma \ref{prop:agent_switch} is that all agents will remain
with their action until one of two events: either there is a weak link
activation or an $\epsilon $-tremble.\footnote{%
An agent can also receive new information from one of her strongly-linked
neighbors, but for this neighbor to change her action in turn requires
either a weak link activation or $\epsilon $-tremble.}\ 

With these results, we are now ready to characterize the BNE action profiles
of the entire network. For this theorem, let us define $\mathfrak{s}_{ij}=1$
if there is a strongly-connected path that links agents $i$ and $j$ (i.e.,
there exists a path of agents $k_{1},\ldots ,k_{K}$ between $i$ and $j$ such
that $S_{ik_{1}}=S_{k_{1}k_{2}}=\ldots S_{k_{K}j}=1$). We also say that a
network is strongly connected if $\mathfrak{s}_{ij}=1$ for all $i,j\in V$.
Finally, we say that a graph is regular if all agents have the same number
of neighbors (and hence $d_{\max }=d_{\min }$).

\begin{theorem}
\label{thm:eq_chara}A BNE always exists. Let $a(t)=[a_{1}(t),a_{2}(t),\cdots 
$ $\cdots ,a_{N}(t)]$ be a BNE action profile for time $t$. Then:

\begin{itemize}
\item If $\tau \leq 1/d_{\max }$, all agents linked by a strongly-connected
path play the same action. That is, for all $i,j\in V$ and all time periods $%
t$,%
\begin{equation*}
\mathfrak{s}_{ij}=1\implies a_{i}(t)=a_{j}(t).
\end{equation*}

\item If $\tau > 1/d_{\min}$ all agents continue to play same action they
were initialized with, i.e., Action $0$, at all time periods $t$.


In particular, if $G$ is also regular ($d_{\max }=d_{\min }$), then we have $%
a_{i}(t)=a_{j}(t)$ for all $i,j\in V$ and all time periods $t$.
\end{itemize}
\end{theorem}

Theorem \ref{thm:eq_chara} greatly simplifies the characterization of any
BNE. Specifically, provided that the degree of local payoff interactions, as
measured by the parameter $\tau $, is not too large, then all
strongly-connected agents and all agents linked via strongly-connected paths
always play the same action. Notably, this is true even when $\tau =0$,
because strong links perfectly transmit information about the underlying
state, creating a powerful force towards all agents playing the same action.
Given this information, agents in a strongly-connected component all have
the same beliefs about which action has greater material payoff.
Consequently, when $\tau =0$, they will all play the same action. The same
conclusion applies when $\tau $ is not too large. In this case, there is an
additional force, which is a desire to match what one's local neighborhood
is doing. This typically reinforces all agents playing the same action in a
strongly-connected component. Nevertheless, the next example shows that when
the parameter $\tau $ is larger than $1/d_{\max }$, the desire to match
one's neighbors can lead to different actions being played in different
parts of a strongly-connected component. The second part of the theorem,
however, shows that even in this case, coordination can be achieved if the
threshold $\tau $ is high enough. However, the downside of such a high
threshold is that even if a node knows that Action 1 has the higher material
payoff, she continues to play Action 0, since all her neighbors are playing
Action 0, and there is more utility in conforming with her neighbors, than
in playing the action with the higher material payoff. 

\begin{description}
\item[Example 2] \label{ex:het_actions}Figure \ref{fig:Example_Heter_Action}
depicts a network in which different actions can be supported among
strongly-connected agents. The figure shows a network with threshold $\tau
=2/5$ where there exists a BNE with different actions within the
strongly-connected component. Intuitively, though strongly-connected, the
network has two different parts and local actions within each part matter
more for payoffs than actions in the other half. This is enough to sustain
an equilibrium in which the left side plays Action 1, while the right side
plays Action 0. 
\begin{figure}[tbp]
\centering
\includegraphics[width=0.45\textwidth]{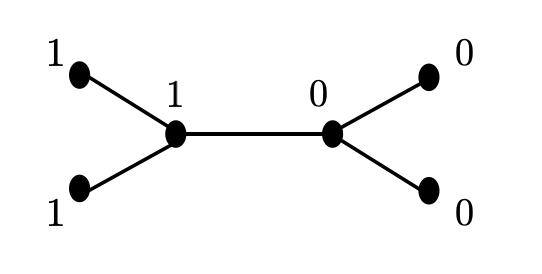}
\caption{A heterogeneous BNE, when $\protect\tau =2/5$.}
\label{fig:Example_Heter_Action}
\end{figure}
This example shows that if we have a high enough threshold, we can maintain
diversity even within a strongly connected network. \hfill
\end{description}



\subsection{Average Welfare\label{sec:gen_analysis}}

In this subsection, we provide a general characterization of average welfare
along a BNE. This characterization builds on defining an embedded Markov
chain over the action profiles of agents in the network.

Note that we use the term \textquotedblleft embedded\textquotedblright\
since we consider the Markov chain in discrete time, although the underlying
learning process is happening in continuous time. In particular, transitions
take place in this Markov chain only at times when there is a payoff shock,
which will be sufficient for us to keep track of long-run average payoffs
(per equation \eqref{def:average_welfare}).\footnote{%
Embedded Markov chains are used in queueing theory, where job arrivals and
departures happen in continuous time, but discrete-time representations
depending on times of job arrival and departure are sometimes more useful
(e.g., \cite{wolff1989stochastic}).}


\begin{definition}[Activation Markov Chain]
\label{def:weak_link_MC} An activation Markov chain (AMC) is an embedded
Markov chain, where the $i^{th}$ transition happens at time $T_{i}$.

\textbf{States of AMC: } The states of this Markov chain are denoted by $%
(P,R,B)$ where:

\begin{itemize}
\item $P \in \{ 0,1 \}^n$ denotes the BNE action profile played by the
agents. 

\item $R \in \{ 0,1 \}$ denotes the action which has the higher reward,
i.e., if we are in time epoch $k$, $R = A_k$.

\item $B\subseteq \mathcal{E}_{w}^{G}$ denotes the set of weak links which
are dormant at the end of an epoch.
\end{itemize}

\textbf{Transition Probabilities of AMC: } The transition probabilities of
this chain are defined as follows: 
\begin{equation*}
\mathbb{P}^{G}(({P}_{l},R_{l},B_{l})|(P_{m},R_{m},B_{m}))=\frac{1}{2}\times 
\mathbb{P}^{G}(P_{l},B_{l}|P_{m},R_{l},B_{m}),
\end{equation*}%
where $\mathbb{P}^{G}(P_{l}B_{l}|P_{m},R_{l},B_{m})$ denotes the probability
that the actions are played according to $P_{l}$, and the weak links in $%
B_{l}$ are dormant, given the action profile is initialized at $P_{m}$, the
weak links in $B_{m}$ are dormant, and the action with the higher material
payoff is $R_{l}$.
\end{definition}

The AMC in Definition \ref{def:weak_link_MC} encapsulates the behavior of
the agents in the network at times of payoff shocks. For example, suppose we
are in state $m$, given by the tuple $(P_{m},R_{m},B_{m})$ at the time of a
payoff shock. By definition, before the arrival of the shock, agents are
playing according to $P_{m}$, and the action with the higher material payoff
was $R_{m}$. Furthermore, the weak links in $B_{m}$ are dormant, meaning
that they are available to potentially become active. This also means that
the weak links in $\mathcal{E}_{w}^{G} \setminus B_m$ are inactive. After
this shock, $R_{l}$ is the action with the higher material payoff.
Thereafter, weak link activation and $\epsilon $-trembles can induce changes
in the action profile of agents. What is particularly convenient in using an
embedded Markov chain is that we do not need to keep track of these
intermediate changes in action profiles. Rather, it is sufficient to focus
on the action profile after all of these changes take place---that is, the
action profile that is being played at the time of the next payoff shock,
which is denoted by $P_{l}$. Furthermore, what information will be
transmitted during an epoch depends on which weak links are dormant, we also
keep track of these in the state $B_{l}$. This also explains why in the
transition probabilities there is a $1/2$: at the time of \ the next payoff
shock, each one of the two actions is the one with the higher material
payoff with probability $1/2$.

In summary, the AMC encapsulates the information about transitions between
action profiles at times of payoff shocks. This is particularly useful,
since from our definition of long-run average payoffs $\mathcal{S}^{{G}}$ in
equation \eqref{def:average_welfare}, it is sufficient to know payoffs at
times of payoff shocks.

The next theorem exploits this feature and characterizes the long-run
average payoffs in terms of the stationary distribution of the AMC.

\begin{theorem}
\label{thm:general_MC_result} For any (weakly)-connected graph ${G}$, the
stationary distribution of the AMC in Definition \ref{def:weak_link_MC},
denoted by $\eta ^{{G}}$, exists. Furthermore, long-run average welfare can
be expressed as a function of this stationary distribution: 
\begin{equation*}
\mathcal{S}^{{G}}=\sum_{q}\eta _{q}^{{G}}f(q),
\end{equation*}%
where $\eta _{q}^{{G}}$ is the stationary probability of state $q$ and $f(q)$
denotes the fraction of agents playing the higher-reward action in state $q$%
, given by 
\begin{align}
f(q)& =f({P}_{q},R_{q},B_{q})  \notag \\
& =\frac{1}{n}\sum_{v\in V}\mathbbm{1}_{{P}_{q}(v)=R_{q}}.  \notag
\end{align}
\end{theorem}


Theorem \ref{thm:general_MC_result} is one of the main results of the paper
and provides a tight characterization of long-run average welfare. In the
rest of the paper, we use this characterization to determine which social
structures achieve a high degree of adaptation and welfare in a changing
environment. This analysis is facilitated by the fact that, as we will see,
the stationary distribution of the AMC is relatively straightforward to
compute in many graphs (including those we will study in our main results in
Theorems \ref{thm:final_1} and \ref{thm:star_is_best_1}). 

We will introduce some additional notation here which will be used
throughout the rest of the paper. Let us define a \emph{conformal state} as
one in which all nodes play the action with the higher material reward and
denote the set of all conformal states by $C$. Similarly, define a \emph{%
diverse state} as one in which not all agents are playing the same\
action---so at least one node is playing Action 0 and at least one node is
playing Action 1. Let us denote the set of diverse states by $D$. We define
the conditional probability of transitioning to a conformal state as: 
\begin{equation*}
p^{G}=\sum_{s\in D}\mathbb{P}^{G}(C|s)\eta _{s}^{G}.
\end{equation*}%
Since $C$ is the set of all possible states where all nodes play the same
action and this action is the one with the higher reward, we have $\mathbb{P}%
^{G}(C|s)=\sum_{B}\mathbb{P}^{G}((1,1,B)|s)+\mathbb{P}^{G}((0,0,B)|s)$.


\section{Adaptation to Change}

\label{sec:under_res}

In this section, we study which network structures are more adaptable to
changing environments---in the sense of generating high long-run average
welfare. In the next subsection, we start with another one of our main
results: in any network without weak links, long-run average welfare is very
low, and in fact only a little bit higher than choosing random actions. Our
next result establishes that an \textit{island network}---where agents are
strongly connected within islands (or components) and islands themselves are
weakly connected---can potentially achieve higher welfare. Finally, we fully
characterize the best network structures from the viewpoint of achieving
long-run adaptation, which turns out to be those that have a star-like
structure, with a large strongly-connected component in the middle, and
weakly-connected leaves providing information to the star component.

\subsection{Low Welfare without Weak Links}

The next theorem is one of our main results and shows that, without weak
links, welfare is very low because society fails to adapt to changes in the
underlying state.

\begin{theorem}[No fast learning without weak links]
\label{thm:final_1} Consider a graph ${G}$ with no weak links. Suppose that $%
\tau \leq 1/d_{\max }$. Then:%
\begin{equation*}
\mathcal{S}^{{G}}\leq \frac{1}{2}+\frac{\epsilon }{2(\lambda +\epsilon )}.
\end{equation*}
Furthermore, when $\tau \geq 1/d_{\min }$, we have $\mathcal{S}^{{G}} = 
\frac{1}{2}$.
\end{theorem}

Theorem \ref{thm:final_1} shows that the long-run average welfare is low and
upper bounded by $1/2+\mathcal{O}(\epsilon )$ in a network without weak
links. Recall that we are interested in economies where $\epsilon $ is very
small (so that trembles are much rarer than payoff shocks). Specifically, as 
$\epsilon \rightarrow 0$, long-run average welfare is no different than an
environment in which no agent has any information about the underlying state
and all players choose their action randomly. Furthermore, if the threshold $%
\tau $ is sufficiently high, no node will change their action and therefore,
the average welfare of such networks will be exactly $1/2$.

While this result may at first appear paradoxical, it is in fact quite
intuitive. Consider a social network in which agents learn the underlying
state at some point and all coordinate in taking the higher-reward action
given this state. Without any weak links and no $\epsilon $-trembles, they
will all continue to play this action, but over time the underlying state
will change, and in the long run, it will only coincide with the initial
state (and thus actions) with probability $1/2$. In this configuration,
long-run average welfare would be exactly $1/2$. A social network without
weak links but with $\epsilon $-trembles can do a little bit better than
this hypothetical situation, because trembles will reveal the underlying
state from time to time, enabling all strongly-connected agents that receive
this information to switch to the higher-reward action. But when $\epsilon $
is small, this adaptation is so slow that it only has a small impact on
long-run average welfare, as formally established in Theorem \ref%
{thm:final_1}.

An immediate implication is that, as we claimed in the Introduction, weak
links are essential for fast learning and adaptation in a changing
environment. The next subsection shows, however, that substituting weak
links for strong ones is not sufficient. The last two subsections then fully
characterize how island networks, connected via weak links, can achieve
higher welfare and what sorts of networks achieve the highest welfare in
this setup.


\subsection{Do Weak Links Necessarily Improve Welfare?}

In this subsection, we compare two simple networks shown in Figure \ref%
{fig:Link_Comparison_Simple} to build some preliminary intuition about the
role of weak links.

\begin{figure}[]
\centering
\includegraphics[width=0.55\textwidth]{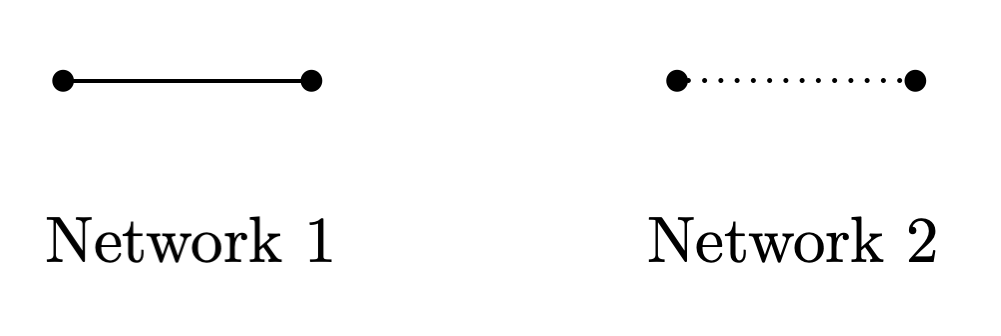}
\caption{Two simple networks}
\label{fig:Link_Comparison_Simple}
\end{figure}

The two networks shown in Figure \ref{fig:Link_Comparison_Simple} have two
agents each. In the first, both agents are connected via a strong link, and
in the second network, the two agents are connected via a weak link. We show
in the Appendix that, in fact, Network $2$ (with the weak link) has a lower
long-run average welfare than Network $1$ (with the strong link). The reason
is that substituting a weak link for a strong one slows down information
transmission and does not alleviate the slow learning problem characterized
in Theorem \ref{thm:final_1}. Instead, long-run adaptation requires weak
links to be \emph{additional} conduits of information, not substitutes for
strong links.


\subsection{Adaptation in Island Networks}

In this subsection, we consider island networks connected via weak links.
While it is hard to characterize the exact welfare for these networks, the
next proposition provides an upper bound on the average welfare for these
networks. For this proposition, recall that $p^{G}$, defined in Section \ref%
{sec:gen_analysis}, corresponds to conditional probability of transitioning
to a conformal state (that is, a state in which all agents play the action
with the higher reward).

\begin{proposition}
\label{prop:island_weak} Consider an island network $G$ with $k$ islands,
each with $m_{c}$ nodes such that $m_{1}\geq m_{2}\geq \cdots \geq m_{k}$.
Furthermore, these islands are connected via weak links. Then, average
welfare can be upper bounded as follows:%
\begin{equation*}
\mathcal{S}^{G}\leq \underbrace{\left( \frac{2p^{G}+\frac{\lambda }{\epsilon 
}p^{G}}{1+2p^{G}+2\frac{\lambda }{\epsilon }p^{G}}\right) }_{(I)}+%
\underbrace{\left( \frac{1}{1+2p^{G}+2\frac{\lambda }{\epsilon }p^{G}}%
\right) }_{(II)}\times \underbrace{\left( \frac{1}{2} + \frac{%
\sum_{i=1}^{\min \{\lceil \gamma /\lambda \rceil ,k\}}m_{i}}{%
\sum_{i=1}^{k}m_{i}}\right) }_{(III)}
\end{equation*}
\end{proposition}

Proposition \ref{prop:island_weak} gives us an upper bound on the average
welfare of island networks with both strong and weak links. Although this
bound is not tight, it is informative about the trade-offs that any network
faces in achieving high average welfare in a changing environment.
Specifically, the right-hand side of Proposition \ref{prop:island_weak}
corresponds to the contribution to average welfare from two set of states
the network may be in: it may be in a conformal state where all nodes play
the action with the higher reward, and this is captured by term \texttt{(I)}%
; or it may be in a diverse state where there are nodes playing both
actions, and the contribution of such states is represented by \texttt{(II)}$%
\times $\texttt{(III)}.

Starting with term \texttt{(I)}, we can see that if $p^{G}$ is large,
welfare in the conformal state will be close to $1/2$ when $\lambda
>>\epsilon $ (which is the case we are focusing on). Intuitively, this
captures the problem that when transition to a conformal state takes place
very rapidly, there will be little adaptation to changes in the underlying
environment. Hence, only networks that have reasonably small values for $%
p^{G}$ can achieve high welfare.

Next, turning to the remaining terms, a small value of $p^{G}$ would ensure
that term \texttt{(II)} is also large, but this has to be coupled with 
\texttt{(III)} being large. This means that either $\gamma /\lambda $ is
large, or $\sum_{i=1}^{d}m_{i}\approx \sum_{i=1}^{k}m_{i}$ for $d<<k$. Yet, $%
\gamma /\lambda $ cannot be large, because this would imply that all weak
links can get activated within an epoch, leading to very large $p^{G}$.
Hence, we must have $\sum_{i=1}^{d}m_{i}\approx \sum_{i=1}^{k}m_{i}$ for $%
d<<k$, which means the largest components of the network must contain most
of the nodes. Hence, to achieve a high upper bound long-run average welfare,
an island network must be such that its largest component contains most of
the agents.

Overall, the upper bound in Proposition \ref{prop:island_weak} highlight the
general forces that contribute to high average welfare. We see in particular
that in order to achieve adaptation in the face of changing environments:

\begin{itemize}
\item a network should be disconnected most of the time, since otherwise it
will generate too much conformity of actions, slowing down learning when the
underlying environment changes. This is achieved in island networks by
having the collection of islands be strongly disconnected. This corresponds
to the requirement that $p^{G}$ should not be too large, which also
encapsulates the requirement that $\gamma /\lambda $ should not be too large;

\item there should nevertheless be information transmission between the
disconnected components at reasonable frequencies. This is achieved in the
island networks by having weak links that are activated at sufficiently high
rates. This corresponds to the requirement that $\gamma /\lambda $ is not
too small;

\item when disconnected, we should still have that a significant fraction of
the agents still play the right action. This is achieved in the island
networks by having each island be strongly connected and weak links carrying
the relevant information to sufficiently many islands. This corresponds to
the requirement that we need the larger islands in the graph to contain most
of the nodes, i.e., $\sum_{i=1}^{\lceil \gamma /\lambda \rceil }m_{i}\approx
\sum_{i=1}^{k}m_{i})$.
\end{itemize}

\subsection{Most Adaptive Networks\label{sec:general_networks}}

\label{sec:star_best}

In the previous subsection, we saw how weakly-connected island networks can
achieve much higher long-run average welfare than our benchmark of networks
without any weak links in Theorem \ref{thm:final_1}. In this subsection, we
turn to the question of whether other networks can even do better and
characterize the best networks from the viewpoint of adaptation to changing
environments. We will see that the same principles highlighted by
Proposition \ref{prop:island_weak} guide the answer to this question. 
Specifically, we will show that a network structure that balances the need
for most agents playing the right action in conformal states and the
imperative of maintaining some diversity for information transmission
achieves the highest feasible payoff.

Anticipating the class of networks that will have these properties, we
define a \textit{star network} with $m$ components and $n$ nodes ($n>m$) as
a network with one component which has $n-m+1$ strongly-connected nodes, and
the other $m-1$ components have size $1$. Furthermore, we suppose that each
of these $m-1$ components has one weak link connecting it to the larger
component of size $n-m+1$. See Figure \ref{fig:Star_Good} for an example of
a star network.

%

\begin{figure}[h!]
\centering
\includegraphics[width=0.55\textwidth]{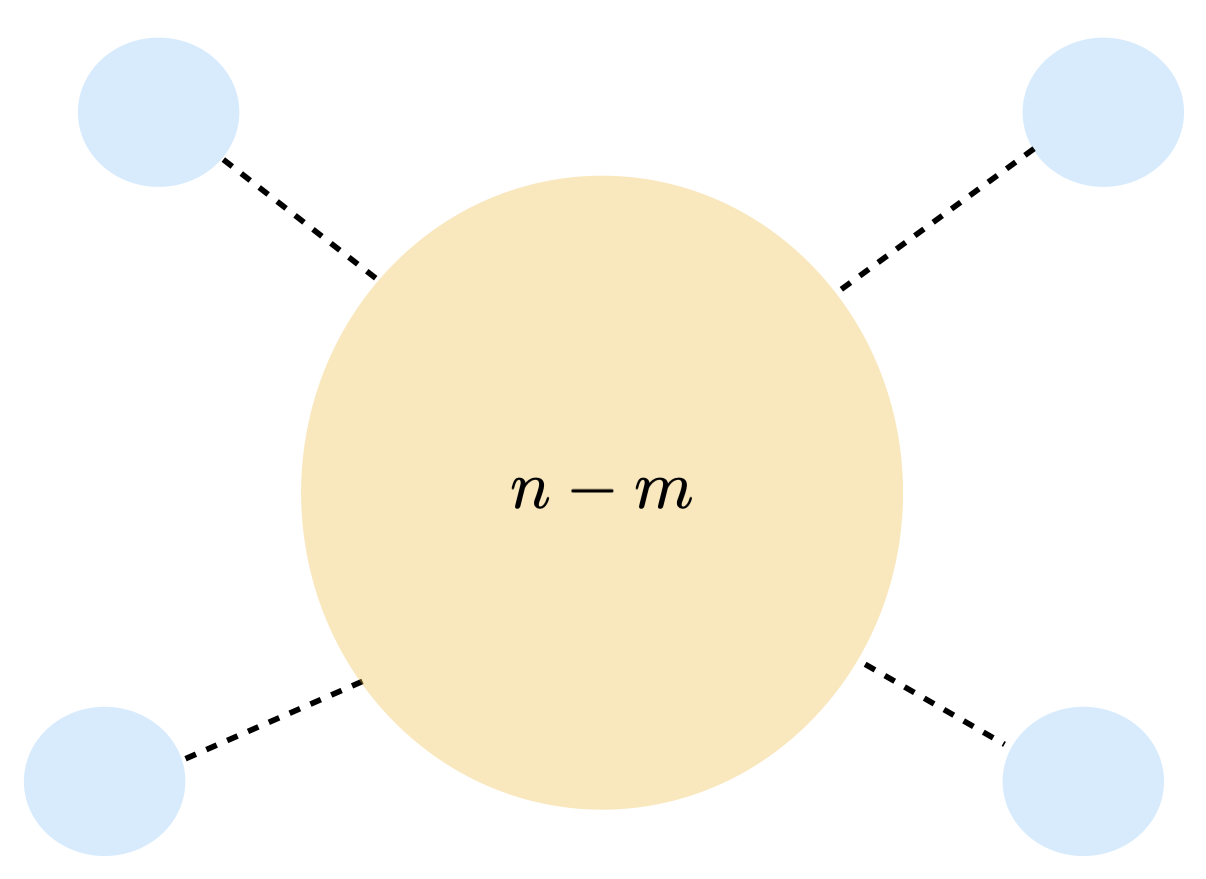}
\caption{A star network with $n-m$ agents which are strongly connected at
the center, and the remaining $m$ agents are connected via weak links to the
core. }
\label{fig:Star_Good}
\end{figure}

The next theorem establishes that the star network, depicted in Figure \ref%
{fig:Star_Good}, achieves the greatest long-run average welfare among all
networks.

\begin{theorem}
\label{thm:star_is_best_1} Given any network $G$ with $n$ nodes, there
exists a star network (shown in Figure \ref{fig:Star_Good}) with the same
number of nodes, that achieves a higher long-run average welfare than $G$,
as $\epsilon \rightarrow 0$ and $\phi \rightarrow 0$ or $\phi \rightarrow
\infty$.

Furthermore, for $\phi <<\lambda <<\gamma $, the average welfare of a star
network approaches $1$ as the number of nodes $n$ and weak links $m$ go to
infinity while $m/n\rightarrow 0$.
\end{theorem}

This theorem establishes two important results. First, a star network (as
defined here) achieves the highest long-run average welfare when
perturbations, given by $\epsilon $ are small, and when transitions of weak
links from inactive to dormant is fast ($\phi $ is large). The these
conditions are both technical and substantive. Substantively, this result
requires $\epsilon -$trembles or mistakes not to be a sufficient source of
(exogenous) diversity. Technically, the limit where $\epsilon \rightarrow 0$
enables us to focus on the case in which all adaptation to a changing
environment comes from agents learning from those who take different
actions. The assumption that either $\phi \rightarrow 0$ or $\phi
\rightarrow \infty $ enables us to focus on the edge cases, where we can
obtain a sharper characterization.

The second part of the theorem shows that, under the sufficient conditions
we impose, average welfare of the star network approaches 1, the highest
feasible payoff in this setting. These conditions require that the number of
nodes to be large relative to the number of weak links (which highlights the
same forces as we emphasized in the previous subsection; we need a
significant fraction of agents to choose the higher-reward action in a
\textquotedblleft diverse" state). For technical reasons, we also send both
the number of nodes and the number of weak links to infinity in this result.
Finally, we also consider the case where $\phi $ is small (though not
necessarily limiting to zero). This condition still ensures that once a weak
link is activated and then becomes inactive, it takes a long time for it to
get out of the inactive state.

We now explain why this property is useful for our result and why the
two-stage activation process for weak links is important for our analysis in
general. First note that without small $\phi $, we can have a situation in
which\ we can start with a network in which all nodes are playing the wrong
action, then the $\epsilon -$ tremble hits one of the agents and the network
moves to a diverse state. Since weak link activation is among agents playing
different actions, it will first pick the agent hit by the $\epsilon -$%
tremble, who will transmit relevant information. Then the next time an
activation takes place, the node hit by the $\epsilon -$tremble has a high
probability of being picked again, and as this happens, the network can
quickly transition to a conformal state again. Introducing the two-stage
activation process, with the backoff period, thus prevents this same node
from being picked in quick succession and helps maintain some amount of
diversity. In other words, this feature enables us to avoid situations where
the network moves to a diverse state and almost immediately moves back to
the conformal state, by having weak links spend a longer duration in the
inactive state.

The proof of this theorem relies on the characterization of average welfare
provided in Theorem \ref{thm:general_MC_result}. We first show that for any
island network $G$ with $k$ components, we can always construct a network in
which one component has size $n-k+1$ and the other $k-1$ components are all
of size $1$. This result thus implies that it is sufficient to restrict
attention to networks that have this special structure. Second, we show that
among all networks with this structure, the star network has the greatest
average welfare. The proof of this step is intuitive and exploits the fact
that the star network achieves the largest component playing the same action
in the middle, while there are sufficient leaves with diversity feeding
information to this middle component. 

\section{Adaptation with Forward Looking Agents}

\label{sec:forward_looking}

We have so far focused on agents that maximize their current (immediate)
payoff, without any weight on future payoffs. If agents are sufficiently
patient, they can themselves engage in experimentation in order to find out
which action is optimal. Although such experimentation issues are important
and interesting, they are beyond the scope of the current paper.\footnote{%
\label{footnote:experimentation_lit}Experimentation over networks is
studied, \emph{inter alia}, in \cite{keller2005strategic}, \cite%
{bonatti2011collaborating, bonatti2017learning}, \cite%
{board2022experimentation}} Nevertheless, it is relevant to investigate
whether forward-looking behavior undoes the main economic forces we have
identified. The next theorem shows that the answer is no, and provided that
agents do not attach too much weight to future payoffs, all of our results
generalize. The bound on the discount factor $\beta $ of the agents depends
on the network structure, and for strongly-connected graphs, it can be
arbitrarily close to 1, as we established next.

The only difference we now consider is that, rather than choosing actions to
maximize current payoffs, as in equation \eqref{eq:action_choice}, each
agent $i\in V$ chooses their action at each time instant to maximize their $%
\beta $-discounted payoff:%
\begin{equation}  \label{eq:ac_disc_t0}
a_{i}(t_{0})=\underset{a\in \{0,1\}}{\func{argmax}}\ \mathbb{E}\left[
\sum_{j=0}^{\infty }\beta ^{j}\mathcal{U}_{i}^{a_{t}(t_{j})}(t)\right] .
\end{equation}
Here, the times $t_i$ are chosen according to a Poisson clock of rate 1.

The next theorem shows that when the discount factor $\beta $ is not too
large, this problem has an identical solution to what we have focused on so
far, thus agents will choose their current-payoff maximizing action and
alter it only at times of new information arrival.

\begin{theorem}
\label{lemma:small_disc_factor}If the discount factor $\beta $ satisfies%
\begin{equation*}
\beta <\frac{\tau d_{\min }}{2+\tau d_{\max }},
\end{equation*}%
then there will be no experimentation and the solution to equation %
\eqref{eq:ac_disc_t0} coincides with 
\begin{equation*}
\underset{a\in \{0,1\}}{\func{argmax}}\ \mathbb{E}_{i,t}[\ \mathcal{U}%
_{i}^{a}(t)].
\end{equation*}%
Consequently, Theorems \ref{thm:eq_chara}-\ref{thm:star_is_best_1} apply
when agents are forward-looking as well.
\end{theorem}

Notice that a higher $\tau $ translates into a higher bound on $\beta $.
This is intuitive. An agent has stronger incentives to conform to her
(strongly-connected) neighbors' actions when local interactions matter more
for payoffs relative to potential gains from individual experimentation. In
fact, returning to the second part of Theorem \ref{thm:eq_chara}, we can see
that for strongly-connected networks or regular networks, We can choose $%
\tau $ sufficiently large to make the bound on $\beta $ sufficiently close
to 1.

Conversely, however, one can also show that there are network and discount
factor combinations for which individual agents would like to experiment,
but we leave further exploration of such situations to future work.


\section{Discussion and Conclusions\label{sec:discussion}}

Diversity facilitates adaptation in both biological and social systems. In
biology, a diverse population is more likely to have sufficient genetic
variation to produce successful strategies against invasions by new species,
food shortages or new climatic conditions. In social systems, diversity can
enable a faster detection of changes in the environment and facilitate
appropriate responses. But, in social networks diversity is even more
difficult to maintain. This is because diversity is essentially a way of
exploring or experimenting with different strategies, while individual
agents tend to have an incentive to exploit the higher payoffs of currently
high-reward actions. When a social network enables fast transmission of
information, however, diversity and sufficient experimentation are
exceedingly difficult to maintain, because all agents learn what is the
currently optimal action and tend to gravitate towards it.

In this paper, we formalized the tension between diversity and exploiting
the high payoff actions. In our model, a collection of Bayesian agents form
beliefs about an underlying state and choose their actions in order to
maximize the sum of the material payoff (coming from matching the underlying
state correctly) and a network payoff (related to coordinating with other
agents one is strongly linked to, such as family, kin group, close friends
or coworkers). Our formal analysis provides a characterization of the
Bayesian-Nash equilibria of this game and provides an explicit formula for
determining long-run average payoffs over the social network.

One of our major results establishes that in a network consisting of just
strong links, long-run average payoffs are approximately the same as
everybody randomly choosing their action. The reason for this is that, with
strong links, any information transmission is fast, so all strongly-linked
components converge to the same action. Once this happens, learning that the
underlying state has changed takes place very slowly, and in the long run,
an initial best action is as likely to be wrong as it is to be right,
yielding approximately the same payoffs as random choice.

Our main results establish that much higher payoffs can be obtained when
strong links and weak links are combined. Weak links, as envisaged by (\cite%
{granovetter1977strength}), involve more infrequent and less tight
interactions than strong links, and in our setting, they transmit
information more slowly (intermittently) than strong links. When weak links
are combined with a structure of strong links that generates multiple
distinct (disconnected) communities, there is room for sufficient diversity.
Each strongly-connected component will still play the same action, but
different components can pursue diverse actions. When the underlying state
changes, one of those components will discover it rapidly. Then this
information will be transmitted to the rest of society via weak links. We
show how a network consisting of strongly-connected islands that are
themselves weakly connected to each other achieves this balance and
consequently much higher payoffs than networks consisting of just strong
links. Social networks that achieve the greatest long-run average welfare
are those that have a star-like structure, whereby a large strongly-linked
component in the middle is fed information from much smaller leaves via weak
links.

Our main results are established for agents who are Bayesian but maximize
current payoffs, which precludes individual incentives for experimentation.
In the final part of the paper, we demonstrate that our results extend to
forward-looking Bayesian agents, provided that their discount factor is less
than a certain threshold, so that they do not have incentives to
individually experiment. This threshold itself crucially depends on local
interactions.

Our paper raises several questions left for future work, and we end with a
brief discussion of a few of these.

\begin{itemize}
\item Incorporating more forward-looking agents who will engage in some
amount of experimentation, though without taking the full social benefits of
diverse actions into account, is one important area of research. This would
amount to combining insights from the emerging literature on experimentation
over networks (such as the works cited in footnote \ref%
{footnote:experimentation_lit}) with those emphasized in this paper, which
focus on issues of adaptation to changing environments.

\item Relatedly, our analysis was simplified by assuming that agents have
identical preferences. Preference heterogeneity generates additional
diversity, and combining this with our overall framework would be another
important direction for theoretical inquiry.

\item In addition to future theoretical work, it would be interesting to
empirically evaluate the linkages between adaptation and diversity in social
systems. As our discussion in the Introduction illustrated, a growing body
of work in biology documents the adaptation benefits of diversity. How the
magnitude and the mechanics of these benefits differ in social settings is a
major question for future research.

\item Lastly, and more broadly, our paper provides one example of how
improved information exchange and communication in a social setting may
generate adverse consequences---in this instance, because it harms diversity
and adaptation. As communication technologies continue to improve at a
breakneck pace, whether unforeseen consequences in terms of diversity,
excessive conformity and adaptation will follow is an important and
multifaceted question that deserves serious study.
\end{itemize}


\newpage 
\bibliographystyle{chicago}
\bibliography{biblio}
{} 


\newpage

\section*{Appendix}

\renewcommand{\theequation}{\mbox{A\arabic{equation}}}

\renewcommand{\thesection}{\mbox{A\arabic{section}}}

\setcounter{equation}{0}

\setcounter{section}{0}

\renewcommand{\thelemma}{\mbox{A\arabic{lemma}}}

\setcounter{lemma}{0}

\renewcommand{\thecorollary}{\mbox{A\arabic{corollary}}}

\setcounter{corollary}{0}

We prove the results stated in the main paper in the appendix. First, we
have the following lemma which will be used in several proofs that follow:

\begin{lemma}
\label{lemma:prob_switch_random} Consider a network with $n$ agents. Assume
that all agents are initialized to play Action $0$, and we have $r_{1}>r_{0}$%
, so that Action $1$ has higher material payoff. Then the probability that
there will be at least one agent in the graph which learns about Action $1$,
before the payoff shock, is given by: 
\begin{equation*}
q=\frac{\epsilon}{\lambda +\epsilon}.
\end{equation*}
\end{lemma}

\begin{proof}
The proof is very straightforward. The random shock happens when the Poisson
clock of rate $\epsilon$ ticks, whereas the payoff shock happens when an
independent Poisson clock of rate $\lambda$ ticks.

Let $X$ denote the random time before the next random switch, and $Y$ denote
the random time before the next payoff shock. Therefore, at least one node
learns about Action 1 before the payoff shock if $X < Y$. However, since
these are independent Poisson clocks, we have that $X \sim \exp(\epsilon)$
and $Y \sim \exp(\lambda)$ and $X$ and $Y$ are independent. Therefore, 
\begin{align*}
q = \mathbb{P} (X < Y) = \frac{\epsilon}{\lambda +\epsilon}.
\end{align*}
which completes the proof.
\end{proof}

\subsection*{Proof of Lemma \protect\ref{lemma:bayes_mono}}

Here, we prove the first part of the lemma. Note that the second part
follows from exactly the same argument as we just redefine belief to be
probability associated with Action 0 being the action with the higher
material payoff. First, note that at time $T_i(t)$, we have $\mu_i(T_i(t)) =
1$. This means that at time $T_i(t)$, node $i$ knows that Action 1 has a
higher material payoff than Action 0 with probability 1.

Now, the belief update of an agent $i$ can be decomposed into two parts:

\begin{itemize}
\item Update due to knowledge of the dynamics of the environment ($%
\mu_i^e(t) $): This captures the belief update due to a payoff shock.

\item Update due to interaction with neighbors ($\mu_i^n(t)$): This captures
the belief update of node $i$, if one of her neighboring nodes has 'learnt'
about the other action.
\end{itemize}

The belief update due to knowledge of the environment evolves as: 
\begin{align}
\mu_i^e(t) &= \left( 1 - e^{-\lambda (t - T_i(t)} \right) \frac{1}{2} +
e^{-\lambda (t - T_i(t))}\mu_i(T_i(t)).  \notag
\end{align}
This can be easily seen as follows: With probability $e^{-\lambda (t -
T_i(t))}$ there has been no payoff shock between times $T_i(t)$ and $t$, and
the belief remains $\mu_i(T_i(t))$. With the remaining probability, there
have been one or more payoff shocks, in which case both actions are equally
likely to be the ones with the higher reward. In this case, the belief is $%
1/2$. The crucial property to note here is that $\mu_i^e(t) \geq 1/2$.

The main hurdle in characterizing the exact belief update $\mu _{i}(t)$, is
in writing the explicit form of the belief $\mu _{i}^{n}(t)$. Note that as
agent $i$ constantly interacts with her neighbors, the fact that there has
been no new information from time $T_{i}(t)$ to $t$ would mean that
potentially, some node maybe explored the other action (either through an $%
\epsilon-$flip or through the activation of a weak link), realized that it
wasn't better, and so continued playing the current higher reward action. In
other words, the fact that no new information about the other action was
received by node $i$ from time $T_{i}(t)$ to $t$ should reinforce the fact
that Action 1 is the higher reward action at time $t$ as well. This is the
only property we need to prove monotonicity of beliefs. More formally, we
can write the belief at time $t$ as: 
\begin{equation*}
\mu _{i}(t)=\int_{T_{i}(t)}^{t}\left[ \mu _{i}^{n}(\tau )\mu
_{i}(T_{i}(t))+(1-\mu _{i}^{n}(\tau ))\mu _{i}^{e}(\tau )\right] d\tau .
\end{equation*}%
The final step of the proof easily follows by noticing that the term inside
the integrand is always greater than $1/2$ (since $0\leq \mu _{i}^{n}(\tau
)\leq 1$ and $\mu_i^e(\tau) \geq 1/2$ for all $\tau \in (T_{i}(t), t) $).

\subsection*{Proof of Lemma \protect\ref{prop:agent_switch}}

We prove this lemma using contradiction. For simplicity assume that $%
a_i(t^-) = 1$, i.e., Node $i$ plays Action 1 at time $t^-$. This means that 
\begin{align}  \label{eq:time_minus}
\mathbb{E}_{i,t^-} [R_1(t-)] + \tau f_i(1,t^-) \geq \mathbb{E}_{i,t^-}
[R_0(t-)] + \tau f_i(0,t^-).
\end{align}

Now, since $t$ is not a time of information arrival, the belief $\mu_i(t)$
will be continuous at time $t$. In particular, this means that we can take
the limit as $t^- \rightarrow t$ in Equation \eqref{eq:time_minus}, and the
sign of the inequality holds. This implies $\mathbb{E}_{i,t} [R_1(t)] + \tau
f_i(1,t) \geq \mathbb{E}_{i,t} [R_0(t)] + \tau f_i(0,t)$ which shows that
Agent $i$ will continue to play Action 1 at time $t$.

Therefore, if node $i$ changes her action at time $t$, this must mean that $%
t $ is a time of new information for node $i$. This completes the proof.

\subsection*{Proof of Theorem \protect\ref{thm:eq_chara}}

Recall that an agent $i$ chooses an action $a$ at time $t$ in order to
maximize $\mathbb{E}_{i,t}[R_a] + \tau f_i(a,t)$. Suppose Agent $i$ plays
Action $0$ and one of her neighbors $j$ plays Action $1$. In this case, both
agents exactly know which action has the higher material payoff. Suppose
Action 1 has the higher payoff. This would mean that Agent $i$ is playing
the action with the lower material payoff.

However, if Agent $i$ switched to Action $1$, her utility would be $1 + \tau
f_i(1,t) > 0 + \tau f_i(0,t)$ since $\tau \leq 1/d_{\min}$. Therefore, Agent 
$i$ would play Action 1 and not 0. This shows that Agent $i$ and all her
neighbors must play the same action. Extending the same argument to all
nodes which are connected to Agent $i$, we have the first part of the
theorem.

The second part follows from the fact that all nodes are initialized to play
the same action. Suppose all nodes are initialized to play Action 0.
Consider Node $i$. All its neighbors are playing Action 0. Now, suppose node 
$i$ has an $\epsilon-$flip at time $t$ and learns that Action 1, in fact has
the higher material payoff. Then, we have $\mathbb{E}_{i,t}[R_1 - R_0] = 1$.
However, since $\tau > 1/d_{\min}$, we have $\tau f_i (0,t) > 1$ (since all
neighbors are also playing Action 0. Therefore, even though node $i$ knows
that Action 1 has the higher material payoff, she continues to play Action
0, since all her neighbors are playing Action 0. This completes the proof.

\subsection*{Proof of Theorem \protect\ref{thm:general_MC_result}}

We first prove that the Markov chain is both irreducible and aperiodic.

\begin{lemma}
\label{lemma:aperiodic_irreducible} The Activation Markov Chain in
Definition \ref{def:weak_link_MC} is both irreducible and aperiodic for any
graph ${G}$
\end{lemma}

\begin{proof}
Consider the state at time $k$ denoted $(P_k, R_k, B_k)$. We have (note that
we drop the superscript $G$ on $\mathbb{P}(\cdot)$ for convenience): 
\begin{align}
\mathbb{P} (({P}_{k+1}, {R}_{k+1}, B_{k+1}) &= (P_k, R_k, B_k) | (P_k, R_k,
B_k) )  \notag \\
&= \frac{1}{2} \times \mathbb{P}(\text{No weak link activation or $\epsilon-$%
flip before shock})  \notag \\
&> 0.  \notag
\end{align}
Therefore, for any state in the Markov Chain, there is a positive
probability of staying in the same state. This shows that the Markov chain
is aperiodic.

Let $\bar{P}_z = \{ z \}^n$ for $z \in \{ 0, 1 \}$ be the action vector
where all agents play the action $z$. Now, consider 2 states, $(P_1, R_1,
B_1)$ and $(P_2, R_2, B_2)$. We show that there is a path of positive
probability between these two states. This can be easily seen as follows
(here $B$ is any subset of weak links): 
\begin{align}
\mathbb{P} ( (\bar{P}_{R_1}, R_1, B_1) | (P_1, R_1, B_1) ) &> 0  \notag \\
\mathbb{P} ( (\bar{P}_{1 - R_1}, 1 - R_1, B_2) | (\bar{P}_{R_1}, R_1, B_1))
&> 0  \notag \\
\mathbb{P} ( (\bar{P}_{1 - R_2}, 1 - R_2, B_2) | (\bar{P}_{1 - R_1}, 1 -
R_1, B_2)) &> 0  \notag \\
\mathbb{P} ( (P_2, R_2, B_2) | (\bar{P}_{1 - R_2}, 1 - R_2, B_2) ) &> 0 
\notag
\end{align}

The first inequality can be seen as follows: 
\begin{align}
& \mathbb{P} ( (\bar{P}_{1 - R_2}, 1 - R_2, B_1) | (P_1, R_1, B_1) )  \notag
\\
& \geq \frac{1}{2}\times \mathbb{P}(\text{Every agent learns the right
action by random flipping before shock})  \notag \\
& >0.  \notag
\end{align}%
The same argument can be used to establish the next inequality. Therefore,
there is a positive probability of moving from any state in this Markov to
chain to any other state. This shows that the Markov Chain in Definition \ref%
{def:weak_link_MC} is irreducible, thereby completing the proof.
\end{proof}

Now, from Lemma \ref{lemma:aperiodic_irreducible}, we know that the Markov
chain is both aperiodic and irreducible. Therefore, since it also has
finitely many states, it has a unique stationary distribution (see for
example \cite{aldous1995reversible}). Thus the stationary distribution $\eta
_{i}^{{G}}$ is well defined. Now, from \cite{aldous1995reversible}, we know
that the limiting behavior of the Markov chain can be characterized by its
ergodic behavior and therefore, we have 
\begin{equation*}
\lim_{k\rightarrow \infty }\mathcal{S}_{k}^{{G}}=\mathbb{E}_{\eta ^{{G}}}[f],
\end{equation*}%
which completes the proof of the theorem.

\subsection*{Proof of Theorem \protect\ref{thm:final_1}}

In the case of a general graph the only possible equilibria are either all
agents play the action with the higher reward or all agents play the lower
reward action (from Theorem \ref{thm:eq_chara}, since $\tau \leq d_{\max}$).
we provide an upper bound for the fraction of agents playing the right
action. We approximate the Markov chain described in Section \ref%
{sec:gen_analysis} with the following 2 state Markov Chain.

\begin{itemize}
\item $\mathtt{G}$ - where all agents play the action with the higher reward.

\item $\mathtt{B}$ - where all agents play the action with the lower reward.
\end{itemize}

The transition probabilities of this 2 state Markov chain is given by: 
\begin{equation*}
\mathbb{P}(\mathtt{G}|\mathtt{B})=\frac{1}{2}(1+q),\qquad \mathbb{P}(\mathtt{%
B}|\mathtt{G})=\frac{1}{2}(1-q).
\end{equation*}%
Here $q$ denotes the probability that some agent will learn about the better
action through a random flip, as derived in Lemma \ref%
{lemma:prob_switch_random}.

\noindent Using these transition probabilities, we have (here $\eta _{i}$
denotes the stationary distribution at state $i$): 
\begin{equation*}
\eta _{\mathtt{B}}=\frac{1}{2}(1-q)\qquad \eta _{\mathtt{G}}=\frac{1}{2}%
(1+q).
\end{equation*}%
Using Theorem \ref{thm:general_MC_result}, this leads to the average welfare
of any connected graph with only strong links given by: 
\begin{equation*}
\mathcal{S}^{{G}}=\frac{1}{2}(1+q).
\end{equation*}%
Now, substituting the value of $q$ from Lemma \ref{lemma:prob_switch_random}%
, we get the final result.

\subsection*{Analysis of Figure \protect\ref{fig:Link_Comparison_Simple}}

For ease of exposition, we consider the limiting behavior when $\phi
\rightarrow \infty $, whereby the weak link is either active or dormant at
all times. From Theorem \ref{thm:final_1}, we know that the average welfare
of Network 1 is given by: 
\begin{equation*}
\mathcal{S}^{{G}_{1}}=\frac{1}{2}\left( 1+\frac{\epsilon }{\epsilon +\lambda 
}\right) .
\end{equation*}

Next, we compute the average welfare of Network 2. There are three possible
states: (i) \texttt{G} (both nodes play the action with higher reward) (ii) 
\texttt{B} (both nodes play the action with lower reward) and (iii) \texttt{M%
} (exactly one node plays the action with higher reward). Consider the
transition probabilities to state $\mathtt{B}$: 
\begin{equation*}
\mathbb{P}(\mathtt{B}|\mathtt{G})=\frac{1}{2}\times (1-q),\qquad \mathbb{P}(%
\mathtt{B}|\mathtt{B})=\frac{1}{2}\times (1-q),\qquad \mathbb{P}(\mathtt{B}|%
\mathtt{M})=0.
\end{equation*}

This shows that the relation between the steady state probabilities of 
\texttt{G} and \texttt{B} is given by: 
\begin{equation*}
\eta _{B}^{G_{2}}=\eta _{G}^{G_{2}}\times \frac{1-q}{1+q}.
\end{equation*}%
%

Now, the average welfare of the second network is given by: 
\begin{align}
\mathcal{S}^{{G}_{2}}& =\eta _{G}^{G_{2}}+\frac{1}{2}\eta _{M}^{G_{2}} 
\notag \\
& =\eta _{G}^{G_{2}}+\frac{1}{2}(1-\eta _{G}^{G_{2}}-\eta _{B}^{G_{2}}). 
\notag
\end{align}

From here, it is easy to see that $G_{2}$ has a lower welfare than $G_{1}$.
Suppose $\eta _{G}^{G_{2}}=\eta _{G}^{G_{1}}-t$ for some $t>0$, we have: 
\begin{equation*}
\mathcal{S}^{{G}_{1}}-\mathcal{S}^{{G}_{2}}=t-\frac{t}{2}\times \left( 1+%
\frac{1-q}{1+q}\right) >0.
\end{equation*}%
This completes the proof. 

\subsection*{Proof of Proposition \protect\ref{prop:island_weak}}

The proof follows from a carefully designed Markov chain to substitute in
Theorem \ref{thm:general_MC_result}. Consider the Markov chain with the
following states:

\begin{itemize}
\item $\mathtt{G}$ - where all agents play the action with the higher reward.

\item $\mathtt{B}$ - where all agents play the action with the lower reward.

\item $\mathtt{M}$ - where there is at least one node which playing Action
0, and at least one node playing Action 1.
\end{itemize}

The transition probabilities between these states is given by (we ignore
terms which involve higher powers of $\epsilon$. Also, we have used Lemma %
\ref{lemma:prob_switch_random} to subsitite the value of $q$):

\begin{align}
\mathbb{P}(\mathtt{M}|\mathtt{G})=\frac{1}{2}\times \frac{\epsilon }{%
\epsilon +\lambda },& \qquad \mathbb{P}(\mathtt{B}|\mathtt{G})=\frac{1}{2}%
\times \frac{\lambda }{\epsilon +\lambda }  \notag \\
\mathbb{P}(\mathtt{M}|\mathtt{B})=\frac{1}{2}\times \frac{\epsilon }{%
\epsilon +\lambda },& \qquad \mathbb{P}(\mathtt{B}|\mathtt{B})=\frac{1}{2}%
\times \frac{\lambda }{\epsilon +\lambda }  \notag \\
\mathbb{P}(\mathtt{M}|\mathtt{M})=1-p^{G},& \qquad \mathbb{P}(\mathtt{B}|%
\mathtt{M})=0.  \notag
\end{align}%
Using these transition probabilities, we get the stationary distribution 
\begin{equation*}
\eta _{M}^{G}=\frac{1}{1+\frac{2\lambda p^{G}}{\epsilon }+2p^{G}},\qquad
\eta _{G}^{G}=\frac{\frac{\lambda p^{G}}{\epsilon }+2p^{G}}{1+\frac{2\lambda
p^{G}}{\epsilon }+2p^{G}}.
\end{equation*}%
Now, in order to derive an upper bound on the average welfare, all that is
left to do is to find an upper bound on the average welfare in the diverse
state (we denote it as $\mathcal{S}_{M}^{G}$).

First, note that if there was no weak link activations, the average welfare
would just be $\approx 1/2$ (since we average out over payoff shocks).
However, on average, there are $\gamma /\lambda $ weak link activations
every epoch. Now, in order to derive an upper bound, we assume that these
weak link activations inform the largest components of the island network.
Therefore, we can upper bound the average welfare in a diverse state as: 
\begin{equation*}
\mathcal{S}_{M}^{G}\leq \frac{1}{2}+\frac{\sum_{i=1}^{\min \{\lceil \gamma
/\lambda \rceil ,k\}}m_{i}}{\sum_{i=1}^{k}m_{i}}.
\end{equation*}%
Therefore, using the fact that the average welfare of the graph can be
written as $\mathcal{S}^{G}=\eta _{G}+\eta _{M}\times \mathcal{S}_{M}^{G}$,
we complete the proof.

\subsection*{Proof of Theorem \protect\ref{thm:star_is_best_1}}

We first prove the part of the theorem which says that a star network is
optimal.

Note that the case where $\phi \rightarrow 0$ follows easily. This case
corresponds to the situation where once a weak link transmits information,
the link becomes inactive and can never be used again. Since the star
network with the middle component having the maximum number of nodes
corresponds to the case where the maximum number of nodes have access to
information flowing through weak links, it has the highest probability of
learning from a weak link, and therefore will have the highest welfare. We
next focus on the case where $\epsilon \rightarrow 0$, and $\phi \rightarrow
\infty$. In particular, this means that all weak links are always dormant or
active, and furthermore, since $\epsilon$ is negligibly small, we only have
to consider its affect when transitioning from a conformal state (since the $%
\epsilon$ shocks are the only way to get out of these states).

We first define a few quantities which will be useful to present our
results. Note that the results are based on Theorem \ref%
{thm:general_MC_result}, for which we need to characterize the Markov Chain
Definition \ref{def:weak_link_MC}. First, we define: 
\begin{equation*}
\mathcal{H}_{n,m}=\{\text{Island networks with $n$ nodes and $m$ components}%
\}.
\end{equation*}

Every graph in $\mathcal{H}_{n,m}$ consists of $m+1$ islands of strongly
connected agents which are connected through weak links. Define 
\begin{equation*}
d_{k}^{G}=\sum_{\text{states $i$ where exactly k components play the right
action}}\eta _{i}^{G}.
\end{equation*}%
In words, $d_{k}^{G}$ represents the steady state probability that exactly $%
k $ components play the action with the higher material payoff. Now, we have
the following property for $d_{k}^{G}$ for all graphs in the class $\mathcal{%
H}_{n,m}$.

\begin{lemma}
\label{lemma:helper_1} For all graphs $G \in \mathcal{H}_{n,m}$, we have: 
\begin{align}
d_k^G = d_k,  \notag
\end{align}
i.e., the distribution of the number of components playing the right action
is the same for all graphs in class $\mathcal{H}_{n,m}$.
\end{lemma}

\begin{proof}
Consider the following Markov Chain representation of the general Markov
chain in Definition \ref{def:weak_link_MC} for a network $G$:

\begin{itemize}
\item States: $s_k$, for $0 \leq k \leq m+1$.
\end{itemize}

Here, the state $s_{k}$ denotes all the action profiles, where exactly $k$
blocks play the right action. Note that for a graph $G$, $d_{k}^{G}$ is the
stationary probability of state $s_{k}$ in this Markov chain. This is an
embedded Markov chain where the state of the system is observed each time
there is a payoff shock.

Now, for a graph $G$ let $\mathbb{P}^G_{s_k}(\cdot)$ denote the transition
probability to any state $s_j$ from state $s_k$ for a graph $G$. We show
that this transition probability is the same for all graphs $G \in \mathcal{H%
}_{n,m}$. This can be easily seen as follows.

\begin{itemize}
\item Transition from state $s_0$ or $s_{m+1}$: This happens initially due
to a random flip with probability $\epsilon$. This is common for all graphs $%
G$.

\item Transition from any other state: This happens due to the activation of
weak links. Since the activation of each weak links, adds one more block to
play the right action, the only factor which determines the transitions are
the number of weak link activations (since we are working in the limit $%
\epsilon \rightarrow 0$ and $\phi \rightarrow \infty$). Since all graphs in $%
\mathcal{H}_{n,m}$ have the same number of components, and each weak link
adds exactly one new block to the number of blocks playing the right action,
we have that this transition probability is also common for all graphs $G
\in \mathcal{H}_{n,m} $.
\end{itemize}

Therefore, since the events which trigger a transition between states is
common for all graphs $G$ and have the same probabilities, the transitions $%
\mathbb{P}^G_{s_k}(\cdot)$ is common for all graphs $G$. Finally, since the
transition probabilities are the same for all graphs $G$, we have that the
final stationary distribution would also be the same, thereby completing the
proof of the lemma.
\end{proof}

Now, note that there are several possible orientations under the constraint
that $k$ out of the $m+1$ blocks play the right action. However, for $k = 0$
and $k = {m+1}$ there is exactly one orientation: All agents play the bad
action, or all agents play the good action respectively. Let $\eta^G_G$ and $%
\eta^G_B$ denote the steady state probability of the all good and all bad
states for graph $G$ respectively. We have the following corollary.

\begin{corollary}
\label{cor:interm_unif_equal} For all graphs $G\in \mathcal{H}_{n,m}$, we
have: 
\begin{equation*}
\eta _{G}^{G}=\eta _{G},\qquad \eta _{B}^{G}=\eta _{B}.
\end{equation*}
\end{corollary}

Next, we define: 
\begin{align}
\mathcal{I}_{n,m}& =\{\text{Island networks with $n$ nodes and $m$
components, where one has}  \notag \\
& \text{ $n-m+1$ nodes and the others just have $1$ node}\}.  \notag
\end{align}%
First, we show that for any island network in $\mathcal{H}_{n,m}$, we can
always find another network in $\mathcal{I}_{n,m}$ which has a higher
average welfare. This is shown in the following lemma:

\begin{lemma}
\label{lemma:helper_2} For any graph $G \in \mathcal{H}_{n,m}$, there exits
a graph $G^{\prime }\in \mathcal{I}_{n,m}$ such that $S^G \leq S^{G^{\prime
}}$.
\end{lemma}

\begin{proof}
Note that the average welfare for a graph with $m$ components, each with $%
k_i $ nodes ($i = 1, 2, \cdots m$) can be written as: 
\begin{align}
\mathcal{S}^G &= \sum_{s} \eta^G_s f(s)  \notag \\
&= \sum_{s} \eta^G_s \left[ \sum_{i=1}^k \mathbbm{1} (\text{component $i$
plays the action with higher reward in state $s$}) \right]  \notag \\
&= \sum_{i = 1}^k \zeta_i^G m_i.  \notag
\end{align}
Here $\zeta_i$ denotes the steady state probability that component $i$ will
be playing the action with the higher material payoff. Let $k^* = \underset{i%
}{\func{argmax}} \ \zeta^G_i$, i.e., $k^*$ is that component in $G$ which
has the highest probability of playing the better action in steady state.

Now, consider another network with the same weak link structure, but all
nodes are in component $k^*$, and all other components have only a single
node. Let this network be denoted by $G^{\prime }$. Note that $G^{\prime
}\in \mathcal{I}_{n,m}$. Furthermore, note that $G^{\prime }$ will have the
same distribution $d_k$ as the graph $G$, and in particular, $\sum \zeta_i^G
= \sum \zeta_i^{G^{\prime }}$ (from Lemma \ref{lemma:helper_1}). Also, since
we are adding more nodes to the component $k^*$, we will have $%
\zeta^{G^{\prime }}_{k^*} \geq \zeta^{G}_{k^*}$. This clearly shows that

\begin{align}
\mathcal{S}^{G}& =\sum_{i=1}^{k}\zeta _{i}^{G}m_{i}  \notag \\
& =\zeta _{k^{\ast }}^{G}m_{k^{\ast }}+\sum_{i\neq k^{\ast }}\zeta
_{i}^{G}m_{i}  \notag \\
& \leq \zeta _{k^{\ast }}^{G}(n-m+1)+\sum_{i\neq k^{\ast }}\zeta _{i}^{G} 
\notag \\
& \leq \zeta _{k^{\ast }}^{G^{\prime }}(n-m+1)+\sum_{i\neq k^{\ast }}\zeta
_{i}^{G^{\prime }}  \notag \\
& =\mathcal{S}^{G^{\prime }}.  \notag
\end{align}%
which completes the proof.
\end{proof}

Therefore, Lemma \ref{lemma:helper_2} tells us that it is enough to restrict
our attention to graphs in $\mathcal{I}_{n,m}$. We refer to the component
with $n-m+1$ nodes as the core of the graph $G$.

We move our attention to a different representation of the markov chain.
Consider the chain shown in Figure

\begin{figure}[h!]
\centering
\includegraphics[width=0.55\textwidth]{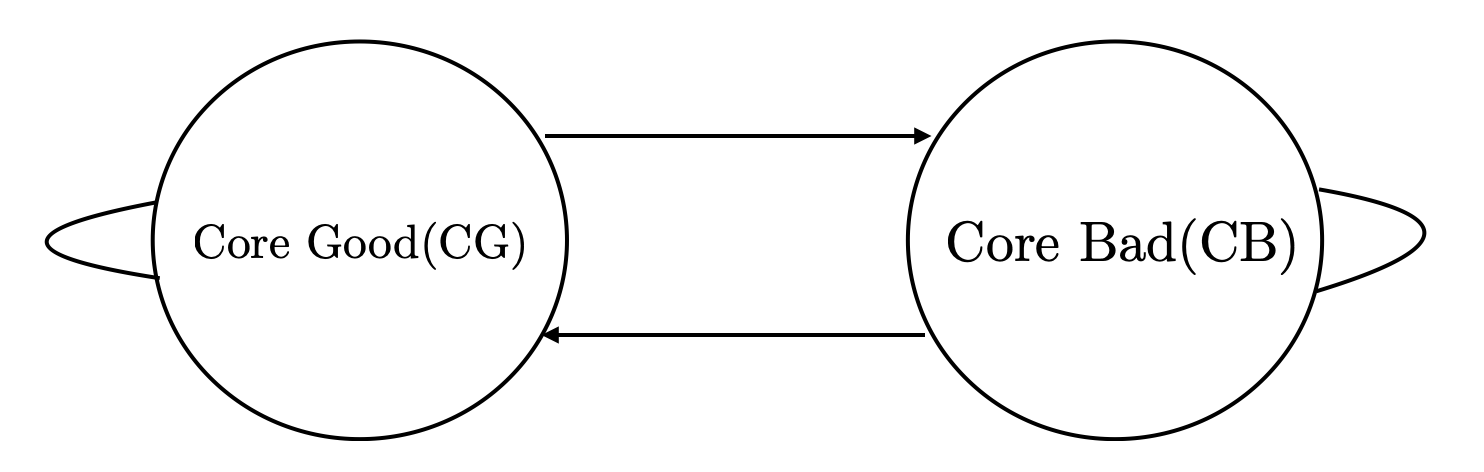}
\caption{Two networks to compare}
\label{fig:Link_Comparison}
\end{figure}

The the two states are the following:

\begin{itemize}
\item Core Good (\texttt{CG}) - These represent the states where the core
component comprising of $n-m + 1$ agents play the right action.

\item Core Bad (\texttt{CB}) - These represent the states where the core
component comprising of $n-m + 1$ agents play the wrong action.
\end{itemize}

Let $\eta _{\mathtt{CG}}^{G}$ and $\eta _{\mathtt{CB}}^{G}$ denote the
stationary distribution of this chain. Note that these stationary
distributions must satisfy: 
\begin{align}
\eta _{\mathtt{CG}}^{G}& =\sum_{\text{states $i$ where core is good}}\eta
_{i}^{G}  \notag \\
\eta _{\mathtt{CB}}^{G}& =\sum_{\text{states $i$ where core is bad}}\eta
_{i}^{G}.  \notag
\end{align}

We have the following crucial lemma which characterizes the behavior of star
networks:

\begin{lemma}
\label{lemma:prefinal} We have: 
\begin{equation*}
\eta _{\mathtt{CG}}^{G_{\text{star}}}\geq \eta _{\mathtt{CG}}^{G},\qquad
\forall \ G\in \mathcal{I}_{n,m}.
\end{equation*}
\end{lemma}

\begin{proof}
First, note from Corollary \ref{cor:interm_unif_equal}, we have that the
state where all blocks play the right action (or all blocks play the wrong
action) have the same probability for all graphs $G \in \mathcal{H}_{n,m}$.

Next, note that from any diverse state where the core is bad, i.e., any
diverse state in $\mathtt{CB}$, the probability of moving to a state in $%
\mathtt{CG}$ is the same for $G_{star}$, i.e., 
\begin{equation*}
\mathbb{P}^{star}(\mathtt{CG}|s)=p_{1}>0,\qquad \forall \ s\in \mathtt{CB.}
\end{equation*}%
This is because in the star network with diverse states, any weak link
activation would let the core know about the correct action and then the new
state will be in $\mathtt{CG}$. The rest of the weak link activations do not
matter. Other events which lead to a state in $\mathtt{CG}$, like the payoff
shock, or $\epsilon $-flip remains the same, independent of the state $s$ in 
$\mathtt{CB}$.

Next, for any other graph $G\in \mathcal{H}_{n,m}$, we have 
\begin{equation*}
\mathbb{P}^{G}(\mathtt{CG}|s)\leq p_{1},\qquad \forall \ s\in \mathtt{CB.}
\end{equation*}%
Note that this is because, from a diverse state, at least one weak link
activation is needed to inform the core about the right action. It might be
possible that more than one weak link activation is needed (depending on the
structure of $G$, as well as the state $s$). Furthermore, the other events
which lead to a state in $\mathtt{CG}$, like the payoff shock, or $\epsilon $%
-flip remains the same, independent of the state $s$ in $\mathtt{CB}$ or the
graph $G$.

These two observations leads to the following inequality for the transition
probabilities for the Markov chain in Figure \ref{fig:Link_Comparison}. 
\begin{equation*}
\mathbb{P}^{star}(\mathtt{CG}|\mathtt{CB})\geq \mathbb{P}^{G}(\mathtt{CG}|%
\mathtt{CB}).
\end{equation*}%
By exactly the same argument, we have 
\begin{equation*}
\mathbb{P}^{star}(\mathtt{CG}|\mathtt{CG})\geq \mathbb{P}^{G}(\mathtt{CG}|%
\mathtt{CG}).
\end{equation*}%
These inequalities on the
transition probabilities give us the desired result.
\end{proof}

Now, we put all these results together to get the final theorem. From Lemma %
\ref{thm:general_MC_result}, we have that the average welfare for a graph $G$
is given by: 
\begin{align}
\sum_{i}\eta _{i}^{G}f(i)& =\frac{1}{n}\sum_{k=1}^{m}\sum_{i_{k}}\eta
_{i_{k}}^{G}\left[ (n-m+1)\mathbbm{1}(\text{core good})+k-\mathbbm{1}(\text{%
core good})\right]  \notag \\
& =\frac{1}{n}\sum_{k=1}^{m}\left[ (n-m-1)\sum_{i_{k}}\eta _{i_{k}}^{G}%
\mathbbm{1}(\text{core good})+k\sum_{i_{k}}\eta _{i_{k}}^{G}\right]  \notag
\\
& =\frac{1}{n}\sum_{k=1}^{m}\left[ (n-m)\sum_{i_{k}}\eta _{i_{k}}^{G}%
\mathbbm{1}(\text{core good})+kd_{k}\right]  \notag \\
& =\frac{1}{n}\left[ (n-m)\sum_{k=1}^{m}\sum_{i_{k}}\eta _{i_{k}}^{G}%
\mathbbm{1}(\text{core good})+\sum_{k=1}^{m}kd_{k}\right]  \notag \\
& =\frac{1}{n}\left[ (n-m)\eta _{\mathtt{CG}}^{G}+\sum_{k=1}^{m}kd_{k}\right]
.  \notag
\end{align}

which is maximized when $\eta_{\mathtt{CG}}^G$ is maximized. Now from Lemma %
\ref{lemma:prefinal}, we have that the star network maximizes this
probability and therefore this completes the proof of the first part of the
theorem.

The second part of the theorem is derived by trying to maximize the
probability that the core component is playing the action with the higher
material payoff. In order to achieve this, not that if we have $\gamma = 
\mathcal{O} (m^{1/2} \lambda)$ and suppose $\phi = \mathcal{O} (m^{-1/4}
\lambda)$, we have that the probability of a weak link being activated in an
epoch $\rightarrow 1$ as $m \rightarrow \infty$. This implies that in every
epoch, the core component will play the better action with probability
approaching 1, as m grows to $\infty$. The average welfare in the limit can
be lower bounded by just the average welfare of the core component, which is
given by $(n-m+1)/m$ which goes to 1, since $m/n \rightarrow 0$. This
completes the proof of the second part of the theorem.

\subsection*{Proof of Theorem \protect\ref{lemma:small_disc_factor}}

\label{proof:disc_factor}

When the Poisson clock (of rate 1) of Agent $i$ ticks , we say that an agent
becomes `active'. When Agent $i$ becomes active at time $t$, let its belief
be $\mu_i(t) = 1/2$. This means that all its neighbors are playing the same
action. For sake of convenience, assume that this is Action $0$. 
In this case, the per time step reward would be $R_0 + \tau f_i(0) > R_0 +
\tau d_{\min}$ where $d_{min}$ is the minimum degree of the graph.

If instead, agent $i$ decides to explore and play Action $1$, the reward
would be $R_1$. However, after exploring, agent $i$ will certainly know
which action is better.

Since the difference between rewards is always $1$, we assume that the
higher reward action has a reward 1, and the lower reward is $0$.

Using this, the total expected reward after exploring can be upper bounded
as: 
\begin{align}
\mathbb{E}[R_{1}]& +\mathbb{E}\left[ \sum_{j=1}^{\infty }\beta ^{j}\left(
R_{a(i,t_{j})}-R_{1-a(i,t_{j})}+\tau
(f_{i}(a(i,t_{j}))-f_{i}(1-a(i,t_{j})))\right) \right]  \notag \\
& \leq \frac{1}{2}+\sum_{j=1}^{\infty }\beta ^{j}\mathbb{E}\left[
R_{a(i,t_{j})}-R_{1-a(i,t_{j})}+\tau (f_{i}(a(i,t_{j}))-f_{i}(1-a(i,t_{j})))%
\right]  \notag \\
& \leq \frac{1}{2}+\sum_{j=1}^{\infty }\beta ^{j}\left( 1+\tau d_{\max
}\right)  \notag \\
& =\frac{1}{2}+\frac{\beta }{1-\beta }(1+\tau d_{\max }).  \notag
\end{align}

On the other hand, if the agent does not explore, the expected sum can be
lower bounded as: 
\begin{align}
\mathbb{E}[R_{1}]& +\mathbb{E}[\sum_{j=1}^{\infty }\beta ^{j}\left(
R_{a(i,t_{j})}-R_{1-a(i,t_{j})}+\tau
(f_{i}(a(i,t_{j}))-f_{i}(1-a(i,t_{j})))\right) ]  \notag \\
& \geq \frac{1}{2}+\tau d_{\min }.  \notag
\end{align}

Now, if 
\begin{equation*}
\frac{1}{2}+\tau d_{\min }>\frac{1}{2}+\frac{\beta }{1-\beta }(1+\tau
d_{\max }),
\end{equation*}%
then the agent will have no incentive to deviate.

Simplifying this inequality, we have: 
\begin{equation*}
\tau d_{\min }>\beta (2+\tau d_{\max }),
\end{equation*}%
which gives us the condition that if: 
\begin{equation*}
\beta <\frac{\tau d_{\min }}{2+\tau d_{\max }},
\end{equation*}%
then the agents will have no incentive to deviate.%

\end{document}